\documentclass{article}
\usepackage{geometry}
\usepackage[nospace,noadjust]{cite}
\usepackage{amsmath,amsthm,amssymb,amsfonts,amsbsy,times}
\usepackage{graphicx}
\usepackage{textcomp}
\usepackage{xcolor}
\usepackage{subfigure}
\usepackage{epsfig}
\usepackage{multirow}
\usepackage{algorithm}
\usepackage{algorithmicx}
\usepackage{algpseudocode}
\usepackage{array}
\usepackage{epstopdf}
\usepackage{color}
\usepackage{diagbox}
\usepackage{bm}
\usepackage{bbm}
\usepackage{mathrsfs}
\usepackage{tcolorbox}
\usepackage{pdfsync}
\usepackage{mathtools}

\usepackage{url}
\usepackage[hidelinks]{hyperref}
\usepackage{paralist}
\usepackage{comment}
\usepackage{multirow}
\usepackage[toc, page]{appendix}

\usepackage{fancyhdr}
\usepackage[capitalize]{cleveref}

\crefrangelabelformat{figure}{#3#1#4--#5#2#6}
\crefrangelabelformat{equation}{(#3#1#4--#5#2#6)}
\crefmultiformat{equation}{Eqs.~(#2#1#3}{,~#2#1#3)}{,~#2#1#3}{,~#2#1#3)}

\newtheorem{defi}{Definition}

\newtheorem{theorem}{Theorem}

\theoremstyle{remark}

\theoremstyle{problem}

\newcommand{\R}{\mathbb{R}}

\newcommand{\C}{\mathbb{C}}

\newcommand{\e}{\begin{equation}}
\newcommand{\ee}{\end{equation}}
\newcommand{\en}{\begin{equation*}}
\newcommand{\een}{\end{equation*}}
\newcommand{\eqn}{\begin{eqnarray}}
\newcommand{\eeqn}{\end{eqnarray}}
\newcommand{\bmat}{\begin{bmatrix}}
\newcommand{\emat}{\end{bmatrix}}

\DeclareMathAlphabet\mathbfcal{OMS}{cmsy}{b}{n}
\renewcommand{\P}[1]{\operatorname{\mathbb{P}}\left(#1\right)}

\newcommand{\imag}{\mathrm{i}}

\newcommand{\vct}[1]{\boldsymbol{#1}}
\newcommand{\mtx}[1]{\boldsymbol{#1}}

\newcommand{\<}{\langle}
\renewcommand{\>}{\rangle}

\newcommand{\trace}{\operatorname{trace}}

\newcommand{\rank}{\operatorname{rank}}

\newcommand{\set}[1]{\mathbb{#1}}

\DeclareMathOperator*{\argmin}{\text{arg~min}}

\newcommand{\wh}{\widehat}

\newcommand{\wt}{\widetilde}
\newcommand{\ol}{\overline}
\newcommand{\nqbit}{n}

\newcommand{\innerprod}[2]{\left\langle #1,  #2 \right\rangle}

\newcommand{\calA}{\mathcal{A}}

\newcommand{\calF}{\mathcal{F}}

\newcommand{\calN}{\mathcal{N}}

\newcommand{\calP}{\mathcal{P}}

\newcommand{\vb}{\vct{b}}

\newcommand{\vf}{\vct{f}}
\newcommand{\vg}{\vct{g}}
\newcommand{\vh}{\vct{h}}

\newcommand{\vp}{\vct{p}}

\newcommand{\vu}{\vct{u}}

\newcommand{\vx}{\vct{x}}

\newcommand{\vpsi}{\vct{\psi}}

\newcommand{\veta}{\vct{\eta}}

\newcommand{\vrho}{\vct{\rho}}
\newcommand{\vsigma}{\vct{\sigma}}

\newcommand{\mA}{\mtx{A}}
\newcommand{\mB}{\mtx{B}}
\newcommand{\mC}{\mtx{C}}

\newcommand{\mF}{\mtx{F}}

\newcommand{\mH}{\mtx{H}}

\newcommand{\mK}{\mtx{K}}

\newcommand{\mP}{\mtx{P}}
\newcommand{\mQ}{\mtx{Q}}
\newcommand{\mR}{\mtx{R}}

\newcommand{\mU}{\mtx{U}}
\newcommand{\mV}{\mtx{V}}
\newcommand{\mW}{\mtx{W}}
\newcommand{\mX}{\mtx{X}}

\newcommand{\mTheta}{\mtx{\Theta}}

\newcommand{\mId}{{\bf I}}

\newcommand{\setF}{\set{F}}

\newcommand{\setX}{\set{X}}

\graphicspath{{./figs/}}

\newlength{\imgwidth}
\newboolean{twoColVersion}
\setboolean{twoColVersion}{false}
\newcommand{\twoCol}[2]{\ifthenelse{\boolean{twoColVersion}} {#1} {#2} }

\definecolor{teal}{RGB}{26,157,150}

\usepackage{mathtools}

\title{\LARGE \bf Structured   Factorization Approaches \\ for     Quantum State Tomography}

\author{Zhen Qin, Joseph M. Lukens, Brian T. Kirby and Zhihui Zhu\thanks{ZQ (e-mail: zhenqin@umich.edu) is with the Michigan Institute for Computational Discovery and Engineering, Department of Electrical Engineering and Computer Science and Department of Statistics, University of Michigan, Ann Arbor, MI 48109 USA; JML (email: jlukens@purdue.edu) is with the Elmore Family School of Electrical and Computer Engineering and Purdue Quantum Science and Engineering Institute, Purdue University, West Lafayette, Indiana 47907, USA, and the Quantum Information Science Section, Oak Ridge National Laboratory, Oak Ridge, Tennessee 37831, USA; BTK (email: brian.t.kirby4.civ@army.mil) is with the DEVCOM Army Research Laboratory, Adelphi, MD 20783, USA and the Tulane University, New Orleans, LA 70118, USA; ZZ (email: zhu.3440@osu.edu) is with the Department of Computer Science and Engineering, The Ohio
State University, Columbus, Ohio 43210, USA.}}

\begin{document}

\maketitle

\begin{abstract}
Since the complexity of quantum state tomography (QST) scales exponentially with system size, exploiting priors such as low-rankness, tensor-network structures, and neural-network representations is essential for scalable QST in terms of sample complexity and parameter complexity. Existing approaches, however, either employ architecture-specific mechanisms to enforce physical validity for mixed states or adopt flexible structural parametrizations without automatically guaranteeing physical validity. In this paper, we introduce a unified framework, termed structured factorization, that builds on Burer–Monteiro-type factorization by parametrizing the density matrix as $\mF\mF^\dagger$, where the factor $\mF$ is constrained to belong to a structured model class. This factorization guarantees physical validity by construction while allowing a broad range of structural priors to be incorporated directly through the choice of the factor space, ranging from the generic Cholesky decomposition to low-rank matrices, matrix product operators, and neural density operators based on multilayer perceptron and transformer architectures. Building on this structured factorization framework, we formulate QST as an optimization problem over the factor space from measurement data. We first develop a unified statistical analysis of the sample complexity of least-squares estimation for a broad class of structured quantum states.
We then propose a projected gradient descent method that operates directly on the factor space and accommodates a wide range of structural parametrizations and reconstruction objectives. To further exploit the geometry of the maximum-likelihood estimation formulation and the constraints on the factors, we derive a power method that yields a step-size-free algorithm with fast convergence, recovering Cover’s method as a special case when the factor is unconstrained.  Numerical experiments demonstrate that the proposed framework, instantiated with low-rank models, tensor-network representations, and neural density operators, enables accurate and scalable QST across a wide range of structured settings.
\end{abstract}


\section{Introduction}

Quantum state tomography (QST) constitutes a cornerstone of quantum information processing and remains the most comprehensive benchmark for the characterization, verification, and validation of quantum devices~\cite{bertrand1987tomographic,vogel1989determination,leonhardt1995quantum,hradil1997quantum,James2001}.
For a composite system of $n$ qudits---each being a $d$-level quantum system, with qubits corresponding to the special case $d=2$---the underlying quantum state is fully described by a density matrix $\vrho \in \mathbb{C}^{d^{n} \times d^{n}}$.
Recovering $\vrho$ in practice necessitates performing quantum measurements on a large ensemble of identically prepared copies of the state.
In the absence of any structural prior, achieving a bounded reconstruction error, measured for instance by the Frobenius norm or trace norm between the reconstructed and true density matrices, provably requires at least $O(d^{2n})$ state copies under independent measurement schemes~\cite{haah2017sample}.
Such exponential sample complexity rapidly renders full QST infeasible as the system size increases, a regime common today with 
state-of-the-art quantum devices now exceeding one hundred qubits~\cite{preskill2018quantum,arute2019quantum,chow2021ibm}.

Driven by the need to bypass the exponential bottleneck of standard QST, a substantial body of work has focused on exploiting low-dimensional structures that intrinsically arise in physically relevant quantum systems.
Among the various structural hypotheses that have been proposed, two paradigms have emerged as particularly influential: low-rank density matrix models and tensor-network representations, the latter most notably exemplified by matrix product operators (MPOs). Low-rankness naturally arises in quantum systems prepared in pure or nearly pure states~\cite{flammia2012quantum,voroninski2013quantum,haah2017sample,kueng2017low,guctua2020fast,francca2021fast,qin2026optimal}; 
specifically, when the density matrix $\vrho$ has rank $r$, the total number of copies required for accurate reconstruction can be reduced to $O(d^{n} r)$, yielding an exponential improvement over unstructured tomography. Despite this dramatic reduction, low-rank assumptions alone remain insufficient in the regime of contemporary large-scale quantum processors. Tensor-network representations, by contrast, offer a qualitatively different route to scalability by explicitly encoding locality and correlation structure, 
efficiently capturing quantum states arising in low-temperature thermal systems~\cite{hartmann2004existence}, one-dimensional spatial systems~\cite{eisert2010}, Hamiltonians with decaying long-range interactions~\cite{pirvu2010matrix}, phase and states generated by noisy intermediate-scale quantum devices~\cite{noh2020efficient}.
A growing line of recent work~\cite{qin2024quantum,qin2024sample,qin2025enhancing,tang2025sketch,votto2025learning,qin2026quantum} demonstrates that the total number of state copies required for accurate reconstruction of tensor-network states can be reduced to polynomial scaling $O(\textup{poly}(n))$ while still guaranteeing bounded recovery error. For a comprehensive review of sample complexity guarantees and theoretical developments in structured QST under compressive measurements, we refer the reader to~\cite{qin2026statistical}.

Beyond tensor-network approaches, neural quantum states and neural density operators (NDOs) have recently emerged as a highly expressive alternative, in which classical neural networks are embedded directly into the parametrization of quantum states. In contrast to tensor-network states, where the representation is defined by a fixed network topology with local tensors obeying specific algebraic structures, neural quantum states adopt a globally shared parametrization in which all amplitudes are generated by the same neural network architecture. While this shared structure enables greater expressive flexibility and the ability to capture complex correlations, it typically comes at the cost of reduced interpretability and less explicit control over entanglement structure.
A broad spectrum of neural architectures has been explored in this context, including restricted Boltzmann machines~\cite{torlai2018neural,torlai2019integrating}, feedforward neural networks~\cite{cai2018approximating}, recurrent neural networks~\cite{morawetz2021u,iouchtchenko2023neural}, transformer-based models~\cite{cha2021attention,ma2023tomography}, convolutional neural networks \cite{fu2024lattice}, and variational autoencoders~\cite{rocchetto2018learning}.
While a rigorous theoretical understanding of the sample complexity required for reconstructing neural quantum states or NDOs remains largely elusive, these models are increasingly regarded as a next-generation representation capable of capturing quantum correlations beyond the expressive limits of conventional tensor-network formalisms.

Alongside ongoing efforts to develop efficient representations of quantum states, a complementary line of research focuses on efficiently determining unknown quantum states from measurement data.  Representative approaches include linear inversion~\cite{fano1957description}, maximum likelihood estimation (MLE)~\cite{hradil1997quantum,vrehavcek2001iterative,James2001}, approximate message passing~\cite{siekierski2025approximate}, cross approximation~\cite{lidiak2022quantum}, Bayesian and region-based inference techniques~\cite{blume2010optimal,granade2016practical,lukens2020practical,blume2012robust,faist2016practical}, as well as least-squares estimation (LSE) and classical machine learning–based schemes~\cite{lohani2020machine,brandao2020fast,zhu2024connection}.
Most of these approaches formulate state reconstruction as an optimization problem over a parametrized family of quantum states. Consequently, the choice of state representation determines not only the expressive power of the model class, but also the computational properties of the resulting estimation problem. For mixed-state tomography, a fundamental challenge is to simultaneously achieve expressive representations and maintain physical validity---i.e., positive semidefiniteness and unit trace. However, many existing representations, including direct density-matrix parametrizations as well as structured models such as MPOs and NDOs, do not automatically enforce these constraints. As a result, physical validity must either be imposed explicitly during optimization or enforced through additional architectural restrictions, which can complicate both theoretical analysis and algorithm design.  Several approaches have been proposed to address this issue. For tensor networks, matrix product density operators (MPDOs)~\cite{verstraete2004matrix} impose structural constraints on the parametrization to ensure positivity by construction, but these constraints are generally sufficient rather than necessary, which may limit representational flexibility. For neural quantum states, mixed-state representations can be constructed through purification-based formulations equipped with autoregressive neural networks~\cite{zhao2024empirical}. Although 
closely related to the Burer--Monteiro factorization, their effective factor dimension is implicitly determined by the purification architecture and ancillary system rather than being explicitly controlled as a model parameter. In other words, low-rankness is not directly parametrized, making it difficult to explicitly enforce a prescribed rank constraint, in contrast to Burer--Monteiro formulations where the factor dimension serves as an explicit and controllable model parameter.

\begin{figure*}[!th]
  \centering
  \includegraphics[width=12.5cm]{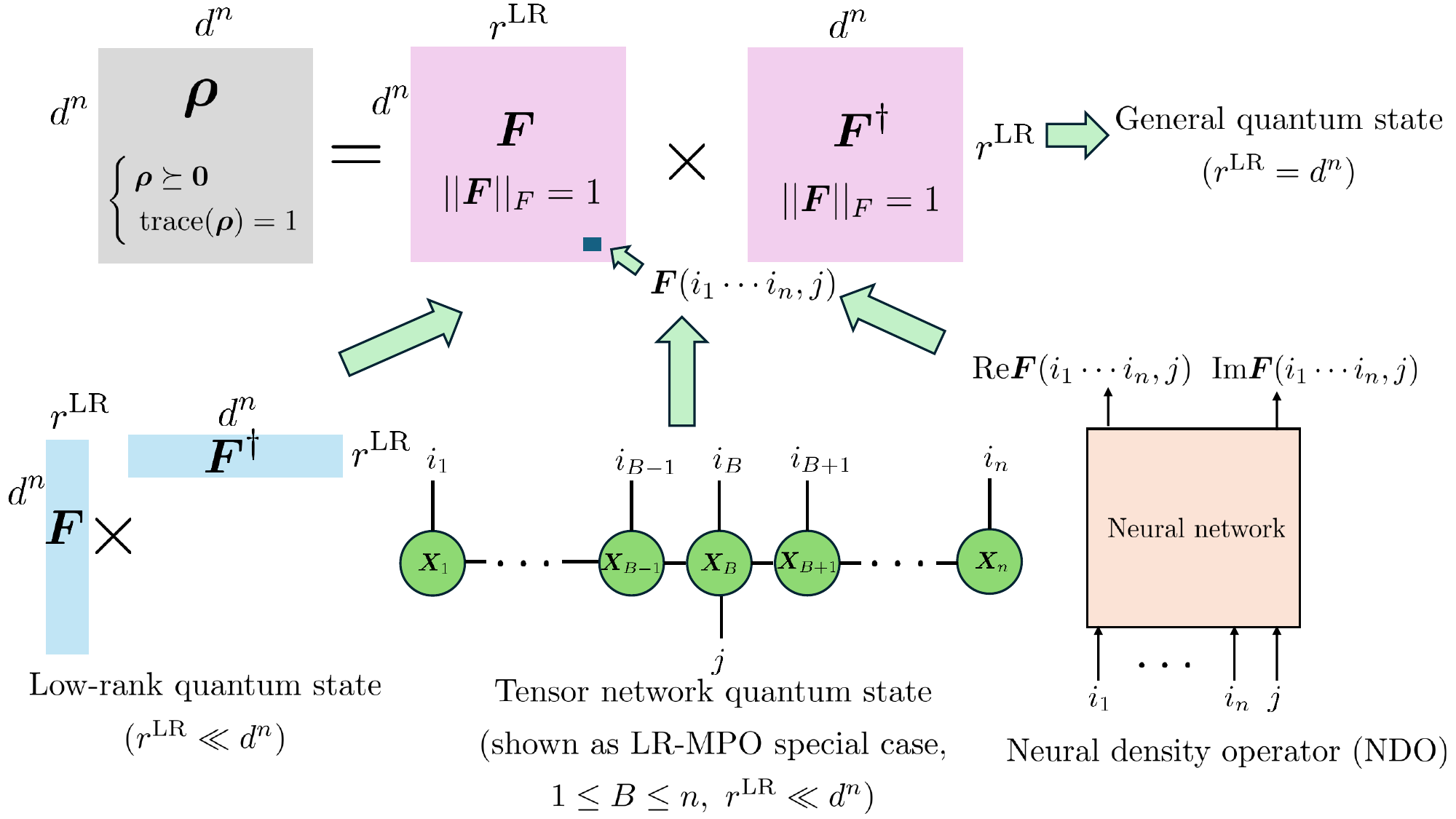}
  \caption{Summary of different physically compatible structured density matrix factorizations, where the definition of $i_1\cdots i_n$ can be found in Eq.~\eqref{the set of MPS states}. }
  \label{fig:summary_structured_states}
\end{figure*}

In this paper, we seek a representation that guarantees physical validity while remaining compatible with a broad range of structural priors, including low-rank, tensor-network, and neural-network parametrizations. Rather than enforcing positivity through architecture-specific constraints, we introduce a {\it physically compatible} parametrization based on the Burer–Monteiro factorization~\cite{burer2003nonlinear,burer2005local}, termed {\it structured factorization},
\begin{equation}
\label{eq:BMF}
    \vrho = \mF \mF^\dagger,
    \qquad
    \mF \in \setF\subset \{\wt\mF\in \C^{d^n \times r^{\textup{LR}}}: r^\textup{LR}\leq d^n,\|\wt\mF\|_F = 1 \},
\end{equation}
which guarantees physical validity by construction. Here $\|\mF\|_F^2 = \trace(\mF^\dagger \mF)$ denotes the Frobenius norm, and $\setF$ is a structured model class that encodes prior assumptions. Depending on the application, $\setF$ may correspond to an unconstrained matrix class, a tensor-network parametrization, or a neural-network-based representation.

We emphasize that the factorization in Eq.~\eqref{eq:BMF} itself is not new.
When $r^{\textup{LR}} = d^n$
and $\mF$ is restricted to be lower triangular with nonnegative diagonal entries, Eq.~\eqref{eq:BMF} reduces to the classical Cholesky factorization \cite{watkins2004fundamentals}, which can represent any positive semidefinite density matrix.
More generally, when $\setF$ is chosen as the class of all matrices in $\C^{d^n\times r^{\textup{LR}}}$, Eq.~\eqref{eq:BMF} reduces to the classical Burer–Monteiro factorization, which has been widely studied in low-rank matrix recovery and low-rank QST \cite{zhu2018global,chi2019nonconvex,wang2024efficient,hsu2024quantum,cai2025online,aditi2025rigorous,qin2026optimal}. The key observation of this work is that the same physically compatible factorization provides a unified framework for incorporating a much broader range of structural assumptions. Rather than viewing Eq.~\eqref{eq:BMF} solely as a low-rank parametrization, we treat it as a physically compatible outer layer that guarantees positivity and trace normalization, while additional structures are encoded through the choice of $\setF$. Building on this perspective, we develop structured parametrizations tailored to low-rank quantum states, including a low-rank matrix product operator (LR-MPO) representation, as well as neural density operator models based on multilayer perceptrons (MLPs) and transformer architectures. Importantly, unlike existing computational frameworks, the proposed approach integrates these structural constraints directly into the parametrization of $\mF$, thereby producing physically valid structured quantum states without post-processing or additional projection steps. An overview of the considered structured state classes and their relationships is summarized in \cref{fig:summary_structured_states}.

We exploit these structured parametrizations for efficient QST by optimizing directly over the factor $\mF$, rather than the density matrix $\vrho$, in fitting to empirical measurements. The resulting optimization formulations are presented later in Eq.~\eqref{General loss function in the FA QST two cases}. Building on this unified structured Burer–Monteiro framework, our technical contributions span both statistical analysis and algorithm design. Our first technical contribution is a unified sample-complexity analysis covering a broad family of structured quantum state classes, including general mixed states, low-rank states, matrix product states (MPSs) for pure states, and the proposed LR-MPO model for mixed states. 
Specifically, we show that, under suitable conditions on the measurement operators, LSE over the factor space achieves near-optimal sample complexity under the Frobenius norm. While a comparable theoretical characterization for MLE remains challenging, we establish that MLE exhibits asymptotically equivalent performance to LSE in the large-sample regime. Moreover, our numerical experiments demonstrate that MLE often provides superior reconstruction accuracy in practical finite-sample settings.

Our second technical contribution is the development of a unified optimization framework tailored to both the structured Burer–Monteiro parametrization and the associated estimation objectives. We first propose a projected gradient descent (PGD) method operating directly on the factor space $\mF\in\setF$, which naturally preserves the normalization constraint and enables efficient optimization across a wide range of structured state classes. 
The proposed PGD framework is agnostic to the specific reconstruction objective and can be applied to a broad class of loss functions, including both LSE and MLE formulations.
Although PGD provides a flexible and broadly applicable optimization strategy, its performance can be sensitive to the choice of step size. To further exploit the structure of the MLE objective, which is widely used in QST due to its statistical optimality and strong empirical performance, we develop a specialized optimization method based on the geometry induced by the Burer–Monteiro parametrization. In particular, the unit-sphere constraint on the factors ($\|\mF\|_F = 1$), together with a multiplicative operator structure naturally arising from the MLE objective, gives rise to a power-method (PM)-type update. Leveraging this observation, we develop an MLE-specific PM that avoids explicit step-size tuning while maintaining fast convergence for a broad class of structured quantum states, including general, low-rank, and tensor-network states for which a canonical factorization is available. 
Extensive numerical experiments validate the effectiveness of the proposed framework and algorithms across a diverse collection of structured QST tasks.

\section{Structure-Preserving Factorization Approaches}
\label{Sec: factorized structured representations}

\subsection{Physical quantum states}

Any valid quantum state must be positive semidefinite and have unit trace---i.e., $\{\vrho\in\C^{d^n\times d^n}: \vrho\succeq {\bm 0}, \trace(\vrho)=1  \}$---but enforcing these structural requirements is often time-consuming in large-scale QST. A common strategy is to employ projection-based approaches~\cite{struchalin2021experimental,kokaew2024bootstrapping,qin2025enhancing}, which enforce positivity and unit trace by explicitly projecting intermediate estimates onto the feasible set. While conceptually straightforward, such methods suffer from two fundamental limitations in high dimensions: first, each projection step typically requires an eigenvalue decomposition or related spectral operation with exponential computational complexity $O(d^{3n} + d^n\log d^n)$; second, they do not naturally generalize to structured parametrizations, such as tensor-network or neural-network representations, making it difficult to incorporate additional structural priors.

These limitations motivate an alternative strategy: instead of repeatedly projecting onto the set of valid density matrices, we reparametrize the density operator in a form that \textit{automatically} satisfies the physical constraints. Such an approach not only eliminates the need for expensive projections but also provides a more principled geometric perspective on the estimation problem. Specifically, for general mixed states, we adopt Burer–Monteiro-type factorization~\cite{burer2003nonlinear,burer2005local}, representing the density matrix as $\vrho = \mF\mF^\dagger$ with $\mF \in \C^{d^n \times d^n}$ and $\|\mF\|_F = 1$, as in Eq.~\eqref{eq:BMF} where $r^{\textup{LR}} = d^n$.  Equivalently, this induces a characterization of the state space via the set of factors\footnote{We chose the label ``simplex'' for this set since the eigenvalues $\{\lambda_k\}$ of all physical states define a standard simplex, i.e., $\lambda_k\geq 0$ and $\sum_k\lambda_k =1$.}
\begin{equation}
    \label{simplex definition of physical quantum states}
    \setF_{\textup{simplex}} = \{ \mF\in\C^{d^n\times d^n}:  \|\mF\|_F = 1  \}.
\end{equation}
Different structured state classes correspond to different choices of the factor class $\setF$. For example, the Cholesky factorization~\cite{watkins2004fundamentals} can be interpreted as a special case obtained by imposing an additional lower-triangular constraint:
\begin{equation}
\label{The definition of Cholesky factorization}
\setF_{\textup{Cholesky}} =
\left\{
\mC\in\C^{d^n\times d^n}:
\|\mC\|_F=1, \mC(k,k)>0 \ \forall k, \
\mC(i,j)=0\ \forall i<j
\right\}.
\end{equation}
This constraint removes the non-uniqueness of the factorization and has been widely adopted in QST~\cite{guctua2012rank,czerwinski2022quantum,koutny2022neural,mohammadisiahroudi2025improvements}. However, Cholesky factorization is primarily a reparametrization of general mixed states and does not fundamentally reduce the representation complexity, as the number of free parameters remains on the order of $d^{2n}$. In the following sections, we describe low-dimensional structured models for the factor $\mF$, simultaneously preserving physical validity and exploiting additional structure in the underlying quantum state.

\subsection{Low-dimensional physical quantum states}

Among the various possibilities of internal state structure, two paradigms are particularly prevalent and operationally significant: (i) low-rank quantum states, characterized by a spectrally compressed density operator, and (ii) tensor-network states---especially those admitting MPS/MPO representations. A number of important state families naturally fall into these categories. For instance,  low-temperature thermal
states~\cite{hartmann2004existence} and phase states whose amplitudes encode low-degree Boolean functions~\cite{ji2018pseudorandom} exhibit effective low-rank structure or permit compact classical descriptions. Likewise, ground states of short-range Hamiltonians and states generated by local quantum dynamics within finite time~\cite{eisert2010} are well captured by MPS/MPO representations, reflecting their constrained entanglement structure and yielding polynomially scalable parametrizations.

\paragraph{Low-rank quantum state.}
When the target state is pure or nearly pure, the density operator exhibits low entropy and can be well-approximated by a low-rank matrix~\cite{kueng2017low,guctua2020fast,francca2021fast,voroninski2013quantum,haah2017sample}.
Formally, we define the set of low-rank quantum states with a rank of at most $r^{\textup{LR}}$ in the factor space as
\begin{eqnarray}
    \label{Definition of low-rank set}
    \setF_{\textup{LR}}= \Big\{ \mF\in\C^{d^n\times r^{\textup{LR}}}: r^\textup{LR}\ll d^n,  \ \|\mF\|_F=1  \Big\},
\end{eqnarray}
which becomes $\setF_{\textup{simplex}}$ when $r^{\textup{LR}} = d^n$.

\paragraph{Matrix product state (MPS).}
Many physically relevant quantum states exhibit additional local or quasi-local structure that cannot be captured solely by rank constraints. In high-dimensional many-body systems, the density operator often admits an efficient tensor-network representation. A particularly powerful form is the MPS, which corresponds to the pure (rank-one) density operator. This representation factorizes the exponentially large vector into a chain of low-order tensors with bounded bond dimension, enabling storage and computation that scale only polynomially in the number of subsystems. Concretely,  let $i_1\cdots i_\nqbit$   denote the row  index\footnote{Specifically, $i_1\cdots i_n$ represents the $(i_1+\sum_{\ell=2}^n d^{\ell-1}(i_\ell-1))$-th row.}, where $i_1,\ldots,i_\nqbit\in [d]$.
With this notation, the feasible set of MPS factors is given by
\begin{eqnarray}
    \label{the set of MPS states}
    \setF_{\textup{MPS}} &\!\!\!\!=\!\!\!\!& \Big\{ \vf\in\C^{d^n\times 1}:\  \|\vf\|_F=1, \vf(i_1 \cdots i_\nqbit) = \mX_1^{i_1}  \cdots \mX_\nqbit^{i_\nqbit}, \nonumber\\
    &\!\!\!\!\!\!\!\!& \mX_\ell^{i_\ell}\in\C^{r_{\ell-1}^{\textup{MPS}}\times r_\ell^{\textup{MPS}}}, \ell\in[n], r_0^{\textup{MPS}}=r_n^{\textup{MPS}}=1  \Big\},
\end{eqnarray}
which implicitly induces the rank-one density operator representation $\vrho = \vf \vf^\dagger$. Here $\{\mX_\ell^{i_\ell} \}_{\ell\in[n]}$ denote the tensor factors in the MPS representation, and $\{r_{\ell}^{\textup{MPS}}\}_{\ell\in[n]}$ are the associated bond dimensions.

\paragraph{Low-rank matrix product operator (LR-MPO).}
While MPS provides an efficient representation for pure quantum states, it does not apply to mixed states. 
A natural generalization is the MPO representation, which is specifically designed for mixed states and, more generally, for operators acting on many-body Hilbert spaces. However, the MPO structure does not naturally integrate with the factorization-based framework. Consequently, obtaining an exact solution within the physical MPO manifold remains challenging, as it is generally difficult to simultaneously satisfy the MPO structural constraints and the simplex constraints. Following~\cite{qin2025enhancing}, one may instead adopt an approximate projection strategy that combines a sequential singular value decomposition---commonly referred to as the tensor-train singular value decomposition (TT-SVD)~\cite{Oseledets11}---with a simplex projection~\cite{chen2011projection} to restore positive semidefiniteness and unit trace. While this two-stage procedure provides a practical means of enforcing physical feasibility, the simplex projection typically perturbs the tensor-network structure of the intermediate estimate. As a consequence, the resulting MPO may no longer strictly satisfy the prescribed bond dimensions.  An alternative is the MPDO~\cite{verstraete2004matrix}, which enforces positive semidefiniteness by construction through locally completely positive factorizations. However, this comes at the cost of increased optimization complexity, as the resulting parametrization is significantly less amenable to efficient learning. These limitations suggest a tradeoff between structural expressivity, physical feasibility, and optimization tractability, motivating the need for a more compact and optimization-friendly representation.

Inspired by the structure of MPS-based density operators and the aforementioned tradeoffs, we propose an LR-MPO representation. Many physically relevant quantum states---such as thermal (Gibbs) states or ground states of gapped local Hamiltonians---simultaneously admit low-rank structure and efficient MPO representations, motivating a unified model that jointly captures both low-rank structure and tensor-network locality. More concretely, we define the set of LR-MPO factors as
\begin{eqnarray}
    \label{the set of LR-MPO states}
    \setF_{\textup{LR-MPO}} &\!\!\!\!=\!\!\!\!& \Big\{ \mF\in\C^{d^n\times r^{\textup{LR}}}:r^\textup{LR}\ll d^n, \ \|\mF\|_F=1, \mF(i_1 \cdots i_\nqbit,j) = \mX_1^{i_1} \cdots \mX_{B}^{i_{B},j} \cdots \mX_\nqbit^{i_\nqbit}, j\in[r^{\textup{LR}}], \nonumber\\
    &\!\!\!\!\!\!\!\!&  \mX_\ell^{i_\ell}\in\C^{r_{\ell-1}^{\textup{LR-MPO}}\times r_\ell^{\textup{LR-MPO}}}, \ell\in[n]\setminus\{B\},  \mX_B^{i_B,j}\in\C^{r_{B-1}^{\textup{LR-MPO}}\times r_B^{\textup{LR-MPO}}},  r_0^{\textup{LR-MPO}}=r_n^{\textup{LR-MPO}}=1 \Big\}.
\end{eqnarray}
Here, $\{\mX_\ell^{i_\ell}\}_{\ell\in[n]\setminus\{B\}}$ together with $\mX_B^{i_B,j}$ denote the tensor factors of the LR-MPO representation, and $\{r_{\ell}^{\textup{LR-MPO}}\}_{\ell=0}^{n}$ are the associated bond dimensions. In contrast to the MPS representation, the additional index $j$ is incorporated into the local tensor $\mX_B^{i_B,j}$ at site $B$. When $r^{\textup{LR}}$ is large, the index $j$ may be decomposed into multiple auxiliary indices, resulting in a block-structured MPO; for notational simplicity, we adopt a single-index notation throughout this paper. The choice of the special site $B$ is arbitrary and does not affect the overall model class.

This LR-MPO family provides a flexible representation that interpolates between purely low-rank models and standard MPO models, allowing one to exploit both spectral compressibility and locality-induced tensor structure. Finally, we note that unlike the conventional MPDO, the LR-MPO expresses the density matrix as a low-rank factorization, where each column of $\mF$ can be viewed as an MPS.
Consequently, writing $\mF=\begin{bmatrix}\vf_1 & \cdots & \vf_{r^{\textup{LR}}}\end{bmatrix}$, each column $\vf_j$ corresponds to an MPS, and the induced density matrix admits the decomposition $\vrho=\sum_{j=1}^{r^{\textup{LR}}}\vf_j\vf_j^\dagger$. This implies that $\vrho = \mF\mF^\dagger$ can be interpreted as a summation of multiple MPSs that share certain tensor factors.

\paragraph{Neural density operator (NDO).}
After discussing intrinsic low-dimensional structures---such as low-rankness, MPS, and LR-MPO---that arise in many-body quantum systems, it is natural to ask whether modern machine-learning-based approaches can exploit similar structural priors. Neural quantum states have demonstrated remarkable expressive power in representing high-dimensional quantum objects, and a comprehensive survey can be found in~\cite{lange2024architectures}. However, most existing architectures are fundamentally tailored to pure (rank-1) states, and extending them to genuinely mixed states, i.e., via NDOs, requires additional modeling choices that introduce significant limitations.

One line of work, exemplified by the Liouville density machine~\cite{kothe2024liouville}, directly parametrizes the density-ket  in Liouville space, i.e., the vectorized form of the density matrix, which in practice amounts to applying an unconstrained neural network directly to a vectorized density matrix. Although this representation is expressive and can capture a wide range of physically relevant states, it does not automatically ensure that the resulting density matrix is positive, and therefore cannot guarantee that the parametrization remains within the physical constraint set. An alternative line of work is purification-based modeling~\cite{torlai2018latent,nomura2021purifying,mellak2024deep}, in which a mixed state is represented as the marginal of a larger pure state supported on an extended Hilbert space. Concretely, auxiliary (ancilla) degrees of freedom are introduced and encoded within the hidden layers of a neural network, which parametrizes a purified wavefunction $\vpsi_a(i)$  on the composite system. The target density operator is then obtained by tracing out the auxiliary subsystem.
Algebraically, this construction induces a positive semidefinite factorization of the density matrix,
\begin{equation}
\vrho(i,j)
\;=\;
\sum_{a} \vpsi_a(i)\,\vpsi_a^*(j)
\;=\;
\sum_{a} \mH(i,a)\mH^*(j,a),
\end{equation}
where the summation index $a$ corresponds to auxiliary degrees of freedom, and the associated factor matrix $\mH$ is implicitly determined by the underlying purification architecture. While this purification-based construction is closely related to low-rank factorizations of the form  $\vrho = \mF \mF^\dagger$, the corresponding rank parameter is embedded in the architecture of the extended system rather than being directly exposed as an optimization variable in the original Hilbert space. This implicit parametrization requires optimization over an enlarged state space together with an auxiliary trace-out operation, introducing additional computational overhead. By contrast, our factorization approach $\vrho = \mF \mF^\dagger$ requires neither artificial rank-1 purification nor relaxation of physical constraints.  To extend this factorized representation to NDOs, we explore two complementary neural architectures. First, an MLP~\cite{cybenko1989approximation} provides the simplest fully connected neural mapping, serving as a baseline for direct function approximation. Second, a transformer network~\cite{vaswani2017attention} leverages self-attention mechanisms to capture complex correlations across qubits, enabling the modeling of nonlocal interactions and higher-order dependencies inherent in many-body quantum systems~\cite{cha2021attention,von2022self,shang2023solving,wu2023nnqs,luo2022autoregressive}. Together, these architectures illustrate a spectrum from straightforward to highly expressive neural parametrizations within the factorization framework. {The key computational building blocks of each network are provided in {Appendix}~\ref{Introduction of MLP and Transformer}}.

We set $\mF$ to have size $d^n \times r^{\textup{LR}}$ as in the sets $\setF_\textup{LR}$ and $\setF_\textup{LR-MPO}$, but now index $\mF$ by a vector of inputs $(i_1 \cdots i_n, j)$ into a neural network. Here, we define the set of NDOs as
\begin{eqnarray}
    \label{the set of NDO states}
    \setF_{\textup{NDO}} = \Big\{ \mF\in\C^{d^n\times r^{\textup{LR}}}:\ \|\mF\|_F=1, \mF \ \textup{represented by a neural network} \Big\}.
\end{eqnarray}
This setup allows a direct and flexible mapping from discrete indices to complex amplitudes, with the network outputting the real and imaginary parts. By formulating the factor in this manner, we can fully leverage the expressive power of modern neural networks while preserving a low-rank representation of the physical quantum state.

\section{Structure-Preserving Estimation Algorithms}
\label{sec: structure preserving factorization approach}

In this section, we introduce a general optimization framework for QST, which can accommodate a variety of reconstruction objectives.
To estimate $\vrho^\star$ of an unknown quantum system, one can perform measurements on a large ensemble of identically prepared copies. The most general class of physically realizable measurements is described by \emph{positive operator-valued measures} (POVMs)~\cite{nielsen2002quantum}, formalized as follows.

\begin{defi}
\label{Definition of POVM}
A POVM is a collection of positive semidefinite matrices $\{\mA_1,\ldots,\mA_K\}$ satisfying
\begin{eqnarray}
\label{The defi of POVM 1}
\sum_{k=1}^K \mA_k = \mId,
\end{eqnarray}
where $\mId$ is the identity matrix.
Each POVM element $\mA_k$ is associated with a possible outcome of a quantum measurement, and the probability $p_k$ of detecting the $k$-th outcome when measuring the density operator $\vrho$ is given by
\begin{eqnarray}
\label{The defi of POVM 2}
p_k = \innerprod{\mA_k}{\vrho} \equiv \trace(\mA_k\vrho),
\end{eqnarray}
where $\sum_{k=1}^Kp_k=1$ following from Eq.~\eqref{The defi of POVM 1} and $\trace(\vrho) = 1$. Repeating the measurement process $M$  times and taking the average of statistically independent outcomes generates the empirical frequencies
\begin{equation}
\wh p_{k} = \frac{f_k}{M}, \  k \in[K],
\label{eq:empirical-prob}\end{equation}
where $f_k$ denotes the number of times the $k$-th outcome is observed, and $[K]\equiv\{1,2,...,K\}$.
\end{defi}
An informationally complete POVM is sufficient to enable full reconstruction of any quantum state, yet  many commonly used POVMs---such as projective rank-one measurements---are not. In this latter case, one typically employs multiple distinct POVMs to characterize of $\vrho$. To formalize this setting, consider $Q$ POVMs indexed by $q \in [Q]$, where each POVM $\{\mA_{q,k}\}_{k\in[K]}$ consists of $K$ PSD operators and is performed with $M$ repetitions, producing empirical frequencies $\wh p_{q,1} \cdots \wh p_{q,K}$.
Assume that the ground truth state admits a structured factorization $\vrho^\star=\mF^\star {\mF^\star}^\dagger$ with $\mF^\star\in\setF$, where $\setF$ denotes a structured model class that encodes prior assumptions, as described in the previous section.
We formulate structured QST as the following optimization problem:
\begin{equation}
\label{General loss function in the FA QST two cases}
    \wh{\mF} =\argmin_{\mF \in \setF} g(\mF) = \begin{dcases}
    \argmin_{\mF \in \setF} \frac{1}{2Q} \sum_{q=1}^{Q}\sum_{k=1}^{K}(\<\mA_{q,k}, \mF\mF^\dagger  \> - \wh p_{q,k}  )^2, &\quad  \text{LSE} \\
    \argmin_{\mF \in \setF}\left( -\frac{1}{Q}\sum_{q=1}^Q \sum_{k=1}^{K}\wh p_{q,k}\log\< \mA_{q,k}, \mF\mF^\dagger\> \right), & \quad  \text{MLE}
    \end{dcases},
\end{equation}
where the objective function $g(\mF)$ specifies the reconstruction criterion. We explicitly consider two widely used choices, LSE and MLE.
LSE reconstructs the density matrix by minimizing the squared deviation between the predicted measurement probabilities and the observed empirical frequencies, offering a computationally tractable approach~\cite{kyrillidis2018provable,guctua2020fast,brandao2020fast,qin2024quantum,zhu2024connection,qin2024sample}. Nevertheless, LSE does not account for the statistical nature of quantum measurements, which are typically multinomial. MLE, on the other hand, provides a statistically principled approach by maximizing the likelihood of the observed measurement outcomes, equivalently minimizing the negative log-likelihood~\cite{hradil1997quantum,vrehavcek2001iterative,James2001}.

\subsection{Sample complexity}

We first formalize the sample complexity of structured QST in Eq. \eqref{General loss function in the FA QST two cases}. We draw upon $\epsilon$-net and covering number theory to
quantify the complexity of the state classes within $\setF$. As a starting point, define the normalized set $\calN = \left\{ \frac{\mF\mF^\dagger}{\|\mF\mF^\dagger\|_F} : \mF \in \setF \right\}$, consisting of all elements of $\{ \mF\mF^\dagger: \mF\in\setF \}$ rescaled to unit Frobenius norm.
A subset $\calN_\epsilon \subset \calN$ is called an $\epsilon$-net (or $\epsilon$-cover)
of $\calN$ if, for every $\frac{\mF\mF^\dagger}{\|\mF\mF^\dagger\|_F} \in \calN$,
there exists some $\frac{\mF'{\mF'}^\dagger}{\|\mF'{\mF'}^\dagger\|_F} \in \calN_\epsilon$
satisfying $\left\|
        \frac{\mF\mF^\dagger}{\|\mF\mF^\dagger\|_F}
        - \frac{\mF'{\mF'}^\dagger}{\|\mF'{\mF'}^\dagger\|_F}
    \right\|_F \leq \epsilon$.
Specifically, for a state class generated by the factorized model $\mF$, we define the covering number $N_\epsilon(\setF)$ as the minimum number of Frobenius-norm balls of radius $\epsilon$ required to cover the collection of density matrices $\{\mF\mF^\dagger:\mF\in\setF\}$. Covering numbers serve as a powerful tool for taming the complexity of large sets:
rather than analyzing every point in the uncountable set $\calN$ directly, one reduces
the problem to a finite collection $\calN_\epsilon$ via the union bound, with each
point in $\calN$ guaranteed to lie within $\epsilon$ of some representative in the cover.

Our analysis departs from the standard approach of bounding $N_\epsilon(\setF)$ directly; instead, we consider the difference class $\{
        \mF_1\mF_1^\dagger - \mF_2\mF_2^\dagger
        : \ \mF_1, \mF_2 \in \setF,\ \mF_1 \neq \mF_2
    \}$ and work with its covering number $N_\epsilon(\ol{\setF})$.
In many settings, $N_\epsilon(\ol{\setF})$ can be upper-bounded by $N_\epsilon^2(\setF)$,
and working with the difference class proves more convenient for the
analysis that follows.
The covering numbers for various classes of quantum states are detailed in Appendix~\ref{proof of recovery error bound} and can be summarized as follows:
\begin{eqnarray}
    \label{covering number for different quantum states1}
    \log N_\epsilon(\ol{\setF}) =  \begin{cases}
    O({d^{2n}}), &\quad  \setF_\textup{simplex}, \setF_\textup{Cholesky} \\
    O(d^{n}r^{\textup{LR}}), &\quad  \setF_\textup{LR}  \\
    O(d\log n\sum_{\ell=1}^n  r_{\ell-1}^{\textup{MPS}} r_\ell^{\textup{MPS}}), &\quad \setF_\textup{MPS}  \\
    O(d^2 \log n\sum_{\ell=1}^n (r_{\ell-1}^{\textup{MPO}})^2 (r_\ell^{\textup{MPO}})^2), &\quad \setF_\textup{LR-MPO}
    \end{cases}.
\end{eqnarray}
These estimates quantify the intrinsic complexity of different parametrizations and will serve as the key input for our subsequent analysis. In particular, they enable a unified characterization of the sample complexity, which will be formally established through the sufficient conditions presented in the following theorem.

\begin{theorem}
\label{tab:recovery_error_structured_POVM}
Suppose $Q$ POVMs $\{\mA_{q,1},\cdots, \mA_{q,K}  \}_{q\in[Q]}$ satisfy
\begin{gather}
    \label{requirements of second order information}
    \sum_{q=1}^{Q}\sum_{k=1}^{K}\<\mA_{q,k}, \vrho - \vrho^\star\>^2 \geq C_1(Q,K)\|\vrho - \vrho^\star\|_F^2,\\
    \label{requirements of third order information}
    \sum_{q=1}^{Q}\sum_{k=1}^{K}\<\mA_{q,k}, \vrho - \vrho^\star\>^2\<\mA_{q,k}, \vrho^\star  \> \leq C_2(Q,K)\|\vrho - \vrho^\star\|_F^2,
\end{gather}
for all $\vrho, \vrho^\star \in \{ \mF\mF^\dagger : \mF \in \setF \}$, where $\setF$ denotes the admissible set of structured quantum states. Using each POVM to measure the state $M$ times yields the empirical distribution $\wh p_{q,1} \cdots \wh p_{q,K}$. With the high probability $1-e^{-\Omega(\log N_\epsilon(\ol{\setF}))}$,  $\wh \vrho = \wh\mF\wh\mF^\dagger$ with $\wh\mF$ being the solution to the constrained LSE in Eq.~\eqref{General loss function in the FA QST two cases} then satisfies
\begin{eqnarray}
\label{final conclusion of recovery error for different states F norm}
\|\wh{\vrho} - \vrho^\star\|_F\leq O\left(\sqrt{\frac{C_2(Q,K)\log N_\epsilon(\ol{\setF})}{C_1^2(Q,K)M}}\right).
\end{eqnarray}
\end{theorem}
The proof is provided in {Appendix}~\ref{proof of recovery error bound}. Intuitively speaking, {Eq.~\eqref{requirements of second order information} guarantees a uniform lower bound on the energy of the error signal under the measurement ensemble, ensuring global identifiability of perturbations in the Frobenius geometry. \Cref{requirements of third order information} further controls higher-order interaction terms between the measurement operators and the ground-truth state, preventing local curvature effects from amplifying stochastic fluctuations. Together, these conditions ensure that the induced least-squares objective exhibits stable local geometry suitable for accurate recovery.}  The specific values of $C_1(Q,K)$ and $C_2(Q,K)$ depend on the measurement ensemble. For spherical $3$-designs~\cite{matthews2009distinguishability,kueng2017low}, they are given explicitly by $C_1(Q,K) = \frac{d^{n}}{K(d^n+1)}$ and $C_2(Q,K) = O\big(\frac{1}{K^2}\big)$. {In contrast, for unitary $3$-designs~\cite{kueng2015qubit,webb2016clifford,zhu2017multiqubit} and Haar-random projective measurements~\cite{qin2024quantum,qin2025enhancing}, Eqs.~\eqref{requirements of second order information} and \eqref{requirements of third order information} hold in expectation with respect to the measurement operators $\mA_{q,k}$}, with  $C_1(Q,K) = \frac{Q}{d^n}$ and $C_2(Q,K) = O\big(\frac{Q}{d^{2n}}\big)$  \cite[Eqs. (S35) and (S36)]{huang2020predicting}. While these expectation-based formulations only hold on average, they are sufficient to establish rigorous high-probability guarantees, since they capture the typical behavior of the measurement ensemble. Notably, although the specific values of $C_1(Q,K)$ and $C_2(Q,K)$ vary across measurement ensembles, their effects are absorbed into the constants of the recovery guarantee. Consequently, the dominant sample-complexity term is governed by the covering complexity of the difference class, $QM\gtrsim \log N_\epsilon(\ol{\setF})$.

According to \Cref{tab:recovery_error_structured_POVM}, the Frobenius-norm recovery error scales proportionally with the intrinsic degrees of freedom of the target quantum state across all considered structural models. By applying the inequality $\|\wh{\vrho} - \vrho^\star\|_1 \leq 2\sqrt{\rank(\vrho^\star)}\|\wh{\vrho} - \vrho^\star\|_F$~\cite{coles2019strong} followed by the Fuchs--van de Graaf inequality $1 - \sqrt{\calF(\wh{\vrho},\vrho^\star)}\leq \frac{1}{2}\|\wh{\vrho} - \vrho^\star\|_1$, we can extend the Frobenius-norm recovery guarantees to corresponding bounds in trace norm and fidelity as follows:
\begin{eqnarray}
\label{final conclusion of recovery error for different states trace norm}
\|\wh{\vrho} - \vrho^\star\|_1 \leq  O\left(\sqrt{\frac{C_2(Q,K)\rank(\vrho^\star)\log N_\epsilon(\ol{\setF})}{C_1^2(Q,K)M}}\right),
\end{eqnarray}
and
\begin{eqnarray}
\label{final conclusion of recovery error for different states fidility}
\calF(\wh{\vrho},\vrho^\star)  \geq    \left[1 - O\left(\sqrt{\frac{C_2(Q,K)\rank(\vrho^\star)\log N_\epsilon(\ol{\setF})}{C_1^2(Q,K)M}}\right) \right]^2.
\end{eqnarray}
Here, the effective rank parameter $\rank(\vrho^\star)$ is specified according to the underlying state structure as
\begin{eqnarray}
\label{rank choice for different states}
\rank(\vrho^\star) =  \begin{cases}
    d^n, &\quad   \mF^\star\in\setF_\textup{simplex}, \setF_\textup{Cholesky} \\
    1, &\quad  \mF^\star\in\setF_\textup{MPS}  \\
    r^{\textup{LR}}, &\quad  \mF^\star\in\setF_\textup{LR},\setF_\textup{LR-MPO}  \\
    \end{cases}.
\end{eqnarray}
Notably, both the trace-norm error and fidelity bounds incur an additional dependence on the effective rank of the ground-truth state $\vrho^\star$, inherited from the norm conversion step.

We emphasize that the present theoretical analysis is tailored to structured matrix-based quantum state models and does not directly extend to NDOs. {This is because different neural network architectures, choices of parameters, and activation functions lead to distinct covering numbers, preventing straightforward generalization.} Nevertheless, we will empirically evaluate the performance of neural parametrizations in simulation to assess their practical behavior. Finally, we emphasize that the above recovery guarantees are derived for constrained LSE. From a statistical perspective, when the number of samples $QM$ is sufficiently large, the negative log-likelihood admits a local quadratic approximation around the true state, under which MLE behaves similarly to LSE and so achieves comparable asymptotic behavior. Outside this regime,  such an approximation no longer holds, and the LSE-based analysis does not directly apply to MLE. Nevertheless, the simulation results presented in \Cref{Sec:simualtion results} indicate that MLE attains lower reconstruction error in practice.

\subsection{Optimization methods}
\subsubsection{Projected gradient descent (PGD)} \label{sec:PGD}
To minimize the loss in Eq.~\eqref{General loss function in the FA QST two cases}, a natural approach is to use PGD on the factorized variable $\mF$, with updates carried out directly in the factor space:
\begin{eqnarray}
    \label{Unified method GD}
 \mF_{t+1} = \calP_{\setF}(\mF_t -   \mu \nabla_{\mF} g(\mF_{t})),
    \end{eqnarray}
where $\calP_{\setF}$ denotes the projection onto the set $\setF$,
$\mu$ is the step size, and $\nabla_{\mF} g(\mF_t)$ denotes the \emph{(Wirtinger) gradient}\footnote{For a comprehensive treatment of Wirtinger calculus and complex matrix analysis, we refer the reader to~\cite[Chapter~3]{zhang2017matrix}.} of the objective function with respect to the complex-valued matrix variable $\mF$. The corresponding expressions for the LSE and MLE losses are given by
 \begin{eqnarray}
    \label{Unified representation of WPGD LSE}
    \nabla_{\mF} g(\mF_t) =
    \begin{dcases}
    \frac{1}{Q}\sum_{q=1}^{Q}\sum_{k=1}^{K}(\<\mA_{q,k}, \mF_t\mF_t^\dagger  \> - \wh p_{q,k}  )\mA_{q,k}\mF_t, &\quad  \text{LSE} \\
     -\frac{1}{Q}\sum_{q=1}^{Q}\sum_{k=1}^{K}\frac{\wh p_{q,k}\mA_{q,k}\mF_t }{\< \mA_{q,k}, \mF_t\mF_t^\dagger\>}, &\quad  \text{MLE}
    \end{dcases}.
    \end{eqnarray}
We note that compared with the LSE loss, the MLE loss can potentially reduce computational cost by exploiting the sparsity of the empirical frequency vector $\wh\vp$. Specifically, whenever some entries of $\{\wh p_{q,k}\}$ are zero---an effect most pronounced when the number of shots per POVM $M$ is much smaller than $d^n$---the corresponding gradient terms $\frac{\wh p_{q,k}\mA_{q,k}\mF_t }{\< \mA_{q,k}, \mF_t\mF_t^\dagger\>}$ vanish and therefore need not be evaluated.

We now characterize the projection operator $\calP_{\setF}$. For low-rank states, the set \(\setF_{\textup{LR}}\) forms a (complex) sphere embedded in \(\C^{d^n \times r^{\textup{LR}}}\), i.e., a smooth manifold defined by a single norm constraint. This simple geometry makes the projection particularly efficient, as it reduces to a normalization step, namely
   \begin{eqnarray}
    \label{Unified method in the FA QST}
    \calP_{\setF_{\textup{LR}}}(\mB)  = \frac{\mB}{\|\mB\|_F},
    \end{eqnarray}
where $\mB\in\C^{d^n \times r^{\textup{LR}}}$. The corresponding per-iteration computational complexities of PGD are
$O(d^{2n} Q K + d^{2n} r^{\textup{LR}})$ and
$O(d^{2n}\textup{NNZ}(\wh\vp) + d^{2n}r^{\textup{LR}})$
for LSE and MLE, respectively, where $\textup{NNZ}(\wh\vp)$ denotes the number of nonzero elements in $\wh\vp$. Since $\setF_{\textup{simplex}}$ corresponds to the full-rank case $r^{\textup{LR}} = d^n$, the projection $\calP_{\setF_{\textup{simplex}}}(\mB)$ coincides with Eq.~\eqref{Unified method in the FA QST}. {For the Cholesky  factorization in Eq.~\eqref{The definition of Cholesky factorization}, one may restrict $\mB$ to the space of lower-triangular matrices, in which case the projection becomes $\calP_{\setF_{\textup{Cholesky}}}(\mB)  = \frac{\textup{Trun}(\mB)}{\|\textup{Trun}(\mB)\|_F},$
where $\textup{Trun}(\cdot)$ zeros-out upper-triangular entries and enforces all diagonal entries to be positive.}

As in the tensor-train decomposition, there is no efficient algorithm for computing the exact projection onto the set $\setF_{\textup{LR-MPO}}$. Instead, we employ the TT-SVD procedure~\cite{Oseledets11}, denoted by $\textup{SVD}^{tt}(\cdot)$, which provides an efficient quasi-optimal projection:
    \begin{eqnarray}
    \label{the normalization step in RPGD LR MPO LSE}
    \calP_{\setF_{\textup{LR-MPO}}}(\mB) = \frac{\textup{SVD}^{tt}(\mB)}{\|\textup{SVD}^{tt}(\mB)\|_F}.
    \end{eqnarray}
Unlike the TT-SVD combined with a simplex projection used in~\cite{qin2025enhancing}, which cannot exactly preserve the bond dimensions, our normalization step guarantees that the bond dimensions remain fixed after the TT-SVD.
From a computational perspective, the per-iteration cost of PGD under the LR-MPO parametrization is
$O(d^{2n} Q K + d^{2n} r^{\textup{LR}} + d^{2n} (\max_\ell r_\ell^{\textup{LR-MPO}})^2)$ for LSE, and
$O(d^{2n}\mathrm{NNZ}(\wh \vp) + d^{2n} r^{\textup{LR}} + d^{2n} (\max_\ell r_\ell^{\textup{LR-MPO}})^2)$ for MLE,
where the last term in both expressions arises from the TT-SVD truncation. Because $\setF_\textup{MPS}$ can be viewed as a special case of $\setF_\textup{MPO}$ with $r^{\textup{LR}}=1$, the projection operator $\calP_{\setF_{\textup{MPS}}}$ is identical to $\calP_{\setF_{\textup{LR-MPO}}}$ with the matrix-valued factorization specialized to the vector case. Table~\ref{tab:qst_complexity structured} summarizes the per-iteration computational complexity of PGD for different structured quantum states based on the LSE and MLE.

\renewcommand{\arraystretch}{1.5}
\begin{table}[t]
\centering
\caption{Per-iteration computational complexity of PGD for various structured quantum states $\mF\in\setF$ based on LSE and MLE.}
\label{tab:qst_complexity structured}
\begin{tabular}{|c|c|c|}
\hline
\textbf{Feasible Set} & \textbf{Loss Function} & \textbf{Computational Complexity} \\
\hline
$\setF_\textup{simplex}$, $\setF_\textup{Cholesky}$ & \begin{tabular}{@{}c@{}} LSE \\ MLE \end{tabular} & \begin{tabular}{@{}c@{}} $O(d^{2n}QK + d^{3n})$  \\ $O(d^{2n}\textup{NNZ}(\wh\vp) + d^{3n})$ \end{tabular}  \\
\hline
$\setF_\textup{LR}$ & \begin{tabular}{@{}c@{}} LSE \\ MLE \end{tabular}  & \begin{tabular}{@{}c@{}} $O(d^{2n}QK + d^{2n}r^{\textup{LR}})$ \\ $O(d^{2n}\textup{NNZ}(\wh\vp) + d^{2n}r^{\textup{LR}})$  \end{tabular} \\
\hline
$\setF_\textup{MPS}$ & \begin{tabular}{@{}c@{}} LSE \\ MLE \end{tabular} & \begin{tabular}{@{}c@{}} $O(d^{2n} Q K + d^{2n}  + d^{2n} (\max_\ell r_\ell^{\textup{MPS}})^2)$  \\ $O(d^{2n}\mathrm{NNZ}(\wh \vp) + d^{2n}  + d^{2n} (\max_\ell r_\ell^{\textup{MPS}})^2 )$  \end{tabular}\\
\hline
$\setF_\textup{LR-MPO}$ & \begin{tabular}{@{}c@{}} LSE \\ MLE \end{tabular} & \begin{tabular}{@{}c@{}} $O(d^{2n} Q K + d^{2n} r^{\textup{LR}} + d^{2n} (\max_\ell r_\ell^{\textup{LR-MPO}})^2)$  \\ $O(d^{2n}\mathrm{NNZ}(\wh \vp) + d^{2n}r^{\textup{LR}} + d^{2n} (\max_\ell r_\ell^{\textup{LR-MPO}})^2 )$  \end{tabular}\\
\hline
\end{tabular}
\end{table}

Unlike the previously considered structured quantum states, for NDOs  the factor matrix is parametrized implicitly via a neural network map $\mH(\cdot)$, i.e., $\mF \coloneqq \frac{\mH(\mTheta)}{\|\mH(\mTheta)\|_F}\in\setF_{\textup{NDO}}$, where $\mTheta$ denotes a collection of trainable weights. Here, the notation $\mH(\mTheta)$ is introduced solely to make the underlying neural network parametrization explicit; throughout the remainder of the paper, we use $\mF$ directly whenever no ambiguity arises. To efficiently learn the unknown weights, we adopt a first-order optimization method and update $\mTheta$ using the Adam optimizer~\cite{kingma2014adam}, which adaptively adjusts the step size for each parameter:
    \begin{eqnarray}
    \label{Unified method in the FA QST NN}
    \mTheta_{t+1} = \textup{Adam}\bigg(\mTheta_t, \nabla_{\mTheta} g\left(\frac{\mH(\mTheta_t)}{\|\mH(\mTheta_t)\|_F}\right)\bigg),   \ \ \ \mF_{t+1}  = \frac{\mH(\mTheta_{t+1})}{\|\mH(\mTheta_{t+1})\|_F}.
    \end{eqnarray}
The detailed update rules of the Adam optimizer are standard and can be found in~\cite{kingma2014adam}. Since the normalization constraint is embedded directly into the forward pass of the model as a reparameterization, the optimization over the network weights $\mTheta$ is completely unconstrained, allowing us to use Adam directly.
In addition, the per-iteration computational cost of the  NDO approach is dominated by the forward and backward passes of the underlying neural network. Consequently, the overall complexity depends on the specific network architecture and the total number of trainable parameters in $\mTheta$, and is therefore not expressed in closed form. However, as common to all methods summarized in \Cref{tab:qst_complexity structured}, its complexity is at least $O(d^{2n})$, corresponding to the size of the full density matrix.

\subsubsection{Power method (PM) for MLE}
While PGD is a standard method for constrained optimization, a critical issue arises in step size selection. For both LSE and MLE, the effective learning rate depends nontrivially on the structure of the state and the number of qudits $n$, and generally must be re-tuned whenever $n$ changes. This sensitivity originates from two sources: the normalization constraint $\sum_{k=1}^{K} \langle \mA_{q,k}, \mF \mF^\dagger \rangle = 1$ $\forall q$, and the finite number of state copies $M$.
For example, under Haar-random projective measurements with $K = d^n$, there exist factor matrices $\mF$ for which each term $\langle \mA_{q,k}, \mF \mF^\dagger \rangle = O(d^{-n})$. Consequently, both the LSE gradient term $\< \mA_{q,k}, \mF \mF^\dagger \> - \wh p_{q,k}$ and the MLE gradient term $\frac{\wh p_{q,k}}{\< \mA_{q,k}, \mF\mF^\dagger\>}$ exhibit scaling that depends sensitively on $n$ and $M$, making step-size selection increasingly challenging as the system size grows.

To circumvent this limitation, we propose an alternative optimization strategy based on the PM~\cite{journee2010generalized}, which iteratively projects the gradient onto the structured constraint set:
    \begin{eqnarray}
    \label{Unified structured representation of PM MLE}
    \mF_{t+1} = \calP_{\setF}(-\nabla_{\mF} g(\mF_t)),
    \end{eqnarray}
where the projection $\calP_{\setF}(\cdot)$ onto the associated set $\setF$ is given in \Cref{sec:PGD}.
The PM-based update can be interpreted as a large-step variant of PGD in a directional sense, where the iterate is obtained by projecting the negative gradient onto the constraint set. The key reason behind this equivalence is that the projection onto the feasible set is scale-invariant, and therefore eliminates the dependence on the gradient magnitude, retaining only its direction. Consequently, its computational complexity matches that of PGD for MLE as reported in \Cref{tab:qst_complexity structured}. In general, such a formulation would be unstable for arbitrary objectives; however, in this setting it is well-defined due to the normalization induced by the projection, and is specifically tailored to the MLE objective. As a result, it does not directly extend to the LSE formulation.

{To understand why the PM update is well-defined in this setting, we draw a connection to the generalized PM of \cite{journee2010generalized}. While the MLE objective $g(\mF)$ is generally nonconvex in the factor $\mF$ due to the quadratic parametrization $\vrho = \mF\mF^\dagger$, it inherits favorable first-order geometry from the convex negative log-likelihood formulation in the density-matrix variable $\vrho$. Specifically, the MLE gradient in Eq.~\eqref{Unified representation of WPGD LSE} takes the form of a data-dependent linear combination of measurement operators, where each term is weighted by the likelihood ratio $\left\{ \frac{\wh p_{q,k}}{\langle \mA_{q,k}, \mF\mF^\dagger\rangle} \right\}_{q,k}$ and applied to the current iterate $\mF$.
This weighting induces a multiplicative correction mechanism that amplifies directions underestimating the observed statistics while suppressing overestimated components.
As a result, the update exhibits nonlinear reinforcement that is structurally analogous to a generalized power iteration. Moreover, the projection ${\calP}_{\setF}(\cdot)$ onto either $\setF_\textup{LR}$ or $\setF_\textup{LR-MPO}$ enforces the Frobenius normalization constraint $\|\mF\|_F = 1$, thereby removing the scale ambiguity of the gradient and stabilizing the infinite-step update. In this sense, the PM iteration may be interpreted as a normalized projected gradient flow induced by the MLE geometry.}

We now specialize to low-rank quantum states, i.e., $\setF = \setF_{\textup{LR}}$.
By \cref{Unified representation of WPGD LSE,Unified method in the FA QST,Unified structured representation of PM MLE},
the PM update becomes
\begin{align}
\label{eq:RrhoR-factor-form}
\mF_{t+1}  = \frac{\displaystyle\sum_{q=1}^{Q}\sum_{k=1}^{K}\frac{\wh p_{q,k}\mA_{q,k}\mF_t }{\< \mA_{q,k}, \mF_t\mF_t^\dagger\>}}{\left\|\displaystyle\sum_{q=1}^{Q}\sum_{k=1}^{K}\frac{\wh p_{q,k}\mA_{q,k}\mF_t }{\< \mA_{q,k}, \mF_t\mF_t^\dagger\>}\right\|_F}.
\end{align}
Let $\vrho_{t} = \mF_{t} \mF_{t}^\dagger$ and define $\mR_t = \sum_{q=1}^{Q}\sum_{k=1}^{K}\frac{\wh p_{q,k}\mA_{q,k}}{\< \mA_{q,k}, \vrho_t\>}$. Plugging these into Eq.~\eqref{eq:RrhoR-factor-form} yields
\begin{align}
\label{eq:RrhoR-factor-form-2}
\mF_{t+1} = \frac{\mR_t \mF_t}{\|\mR_t\mF_t\|_F},
\end{align}
which induces the following update in terms of the density matrix:
\begin{align}
\label{eq:RrhoR}
\vrho_{t+1} = \mF_{t+1}\mF_{t+1}^\dagger  = \frac{\mR_t \vrho_{t} \mR_t}{\trace(\mR_t \vrho_{t} \mR_t)}.
\end{align}
When $\mF$ is full-rank ($\setF = \setF_{\textup{Simplex}}$), this update recovers the iterative MLE scheme in~\cite{lvovsky2004iterative}, {which is also known in the literature as Cover’s method}. To our knowledge, this provides a new interpretation of that algorithm from a PM perspective. Furthermore, the factorized derivation in Eq.~\eqref{eq:RrhoR-factor-form-2} naturally extends the iterative MLE procedure to structured settings such as low-rank states, and more broadly the PM in Eq.~\eqref{Unified structured representation of PM MLE} to other representations including MPS and MPO with corresponding projection. In the next section, we demonstrate the performance of the PM for recovering various quantum states.

\section{Simulation Results}
\label{Sec:simualtion results}

In this section, we perform numerical QST experiments using Haar-random projective measurements\footnote{Haar-random projective measurements are generated by preparing $Q$ independent Haar-random unitaries $\{\mU_q\}_{q\in[Q]}$ and forming the associated POVMs $\{\mU_q(:,k)  \mU_q^\dagger(:,k),k \in [ d^n] \}$ for each $q \in [Q]$.} to compare reconstruction methods based on LSE and MLE.  The evaluated methods are labeled according to the convention [Structure]-[Algorithm]-[Loss function]: LR-PGD-LSE, LR-PGD-MLE, LR-PM-MLE, LR-MPO-PGD-LSE, LR-MPO-PGD-MLE, LR-MPO-PM-MLE, MLP-PGD-LSE, Transformer-PGD-LSE, MLP-PGD-MLE, and Transformer-PGD-MLE.  For initialization of PGD and PM in the low-rank and LR-MPO models, we set $\mF_0 = \frac{ \mA + \imag\mB}{\|\mA + \imag\mB\|_F}\in \C^{d^n\times r^{\textup{LR}}}$, where the entries of $\mA$ and $\mB$ are independent and identically distributed samples drawn from the standard normal distribution.  For the MLP, all linear layer weights are initialized independently from a normal distribution with zero mean and standard deviation $0.1$, and all biases are initialized to zero. For the transformer, all linear layers—including the query, key, value, and output projection matrices in self-attention, as well as the feedforward sublayers—are initialized using the Xavier uniform distribution with a gain factor of $0.1$.
For each configuration, we conduct 10 Monte Carlo tomographic experiments in which each Haar measurement and result are sampled at random; we average over all 10 trials to report the results. To evaluate the performance of the estimated state $\wh{\vrho}$ against the ground truth $\vrho^\star$, we employ three error metrics: the normalized mean squared error $\|\wh{\vrho} - \vrho^\star\|_F^2 / \|\vrho^\star\|_F^2$, the trace norm $\|\wh{\vrho} - \vrho^\star\|_1$, and the fidelity $\calF(\wh{\vrho}, \vrho^\star)$. Since the Frobenius norm $\|\vrho^\star\|_F$ varies across quantum states with different ranks or bond dimensions, we use a normalized version of $\|\wh{\vrho} - \vrho^\star\|_F^2$  to ensure consistent cross-state comparisons. 

To evaluate the performance of the aforementioned reconstruction methods under diverse structural properties, we consider three representative classes of quantum states:
\begin{itemize}
    \item \textbf{Thermal state} generated from the 1D quantum Ising model:
    \begin{equation}
    \vrho^\star_{\text{thermal}} = \frac{e^{-\mH/T}}{\trace(e^{-\mH/T})},
    \qquad
    \mH = \sum_{j=1}^{n-1}\vsigma_{z,j}\vsigma_{z,j+1} + \sum_{j=1}^{n}\vsigma_{x,j}
    \end{equation}
    where $\vsigma_{a,j} = \mId_{2^{j-1}} \otimes \vsigma_a \otimes \mId_{2^{n-j}},\ a \in\{x,z\}$ with
    \(\vsigma_x = \begin{bmatrix}0 & 1\\1 & 0\end{bmatrix},\quad \vsigma_z = \begin{bmatrix}1 & 0\\0 & -1\end{bmatrix}.\)
    The temperature \(T\) controls the effective rank and, for MPO representations, the bond dimension of \(\vrho^\star_{\text{thermal}}\). Lower \(T\) leads to states closer to the ground state with smaller rank and bond dimension, whereas higher \(T\) increases rank and bond dimension.
    \item \textbf{Greenberger--Horne--Zeilinger (GHZ) state}:
    \begin{equation}
    \vrho^\star_{\text{GHZ}} = \vg \vg^\dagger, \qquad
    \vg = \frac{1}{\sqrt{2}}\begin{bmatrix}1 & 0 & \cdots & 0 & 1\end{bmatrix}^\top \in \mathbb{R}^{2^n \times 1}.
    \end{equation}
    \item \textbf{Zero-texture (Fourier) state}\footnote{The zero-texture state  corresponds to a special case of the Fourier state~\cite{parisio2024quantum}.}:
    \begin{equation}
    \vrho^\star_{\text{zero}} = \vu \vu^\dagger, \qquad
    \vu = \frac{1}{\sqrt{2^n}} \begin{bmatrix}1 & 1 & \cdots & 1\end{bmatrix}^\top \in \mathbb{R}^{2^n \times 1}.
    \end{equation}
\end{itemize}
For the GHZ and zero-texture states, the rank is one and the bond dimension for an MPO representation is minimal. We emphasize that due to the dependence on the system size \(n\) and temperature \(T\), the thermal state does not admit a fixed effective rank or MPO bond dimension. In contrast, the GHZ and zero-texture states admit exact MPO representations with constant bond dimensions ($4$ and $2$, respectively), independent of the system size $n$.

In the first set of experiments, we compare the reconstruction performance of different methods 
with $Q=100$ and $M \in \{5,100,1000\}$.  All experiments are conducted on systems with $n=6$ qubits, and for the thermal state we fix the temperature parameter to $T=0.2$. For $\setF_\textup{LR}$ and $\setF_\textup{LR-MPO}$ cases, the rank and bond dimensions are chosen according to the underlying structure of the target states. For the thermal state, we set $r^{\textup{LR}}=2$ and apply TT-SVD to the reconstructed estimator, with an error tolerance of $10^{-14}$ and the position $B=\lceil n/2\rceil$ in Eq.~\eqref{the set of LR-MPO states}, to adaptively select the MPO bond dimensions. For the rank-one GHZ and zero-texture states, we set $r^{\textup{LR}}=1$ and use the same error tolerance. For $\setF_\textup{NDO}$, we consider both MLPs and transformer-based architectures. The MLP consists of $L=2$ layers with ReLU activations, where the hidden dimension $d_{\textup{MLP}}$ is set to $24$, $8$, and $16$ for the thermal, GHZ, and zero-texture states, respectively. The transformer comprises $L=2$ stacked layers with model dimension $d_{\textup{trans}}$ of $24$, $8$, and $16$ for the three states, respectively.
We use a fixed attention window (chunk) size $D_{\textup{token}}=4$: the full sequence of $d^n r^{\textup{LR}}$ index tuples (under a fixed lexicographic ordering) is partitioned into contiguous blocks of size $D_{\textup{token}}$, and self-attention is computed independently within each block.
Multi-head self-attention with $M_{\textup{attn}}=2$ heads is employed, and each encoder layer includes a position-wise feed-forward network with an expansion factor of four. Detailed neural network architectures are provided in {Appendix}~\ref{Introduction of MLP and Transformer}. No dropout is applied in any of the neural architectures. The specific step sizes and iteration counts used by the different reconstruction methods are summarized in \Cref{parameters of different methods}.
This table is provided to document the hyperparameters required for convergence and to ensure reproducibility of the reported results. We note, however, that iteration counts should not be interpreted as a direct measure of computational efficiency, since the computational cost per iteration varies substantially across different classes of methods (see \Cref{tab:qst_complexity structured}). Moreover, the reconstruction procedures do not assume access to the underlying quantum state; all hyperparameters are selected according to standard benchmark settings and coarse structural priors rather than oracle knowledge of the ground truth.

From \Cref{thermal_summary_full,ghz_summary_corrected,w_summary_corrected}, several overarching trends emerge. First, MLE consistently achieves superior reconstruction accuracy across most regimes, reflecting its principled treatment of  statistical fluctuations.  Second, incorporating MPO structure into low-rank reconstruction leads to an improvement over conventional low-rank methods, underscoring the benefit of exploiting physically motivated tensor-network representations to suppress noise.
Third, NDO models exhibit distinct advantages in selected scenarios---most notably for GHZ and zero-texture states---although their convergence typically requires a larger number of iterations, highlighting a tradeoff between expressive power and optimization efficiency for highly structured quantum states.

\begin{table}[!t]
\renewcommand{\arraystretch}{1}
\centering
\footnotesize
\caption{Experiment 1---Step sizes and iteration counts required for different methods across quantum states.}
\label{parameters of different methods}
\begin{tabular}{|c||c|c|c|}
\hline
\multirow{1}{*}{Methods}
 & \multicolumn{1}{c|}{{Thermal state} }
 & \multicolumn{1}{c|}{{GHZ state}}
 & \multicolumn{1}{c|}{Zero-texture state} \\  \hline

\multicolumn{4}{|c|}{Step sizes / Iteration counts} \\ \hline
LR-PGD-LSE
 & $10$ / $500$ & $40$ / $100$ & $40$ / $100$  \\
LR-PM-MLE
 & -- / $100$ & -- / $100$ & -- / $100$ \\ \hline
LR-MPO-PGD-LSE
 & $40$ / $200$ & $40$ / $100$ & $40$ / $100$   \\
LR-MPO-PM-MLE
 & -- / $100$ & -- / $100$ & -- / $100$  \\ \hline
MLP-PGD-LSE
 & $10^{-3}$ / $500$ & $10^{-3}$ / $500$ & $10^{-3}$ / $500$  \\
Transformer-PGD-LSE
 & $10^{-3}$ / $500$ & $10^{-3}$ / $500$ & $10^{-3}$ / $500$  \\
MLP-PGD-MLE
 & $10^{-2}$ / $500$ & $10^{-2}$ / $500$ & $10^{-2}$ / $500$  \\
Transformer-PGD-MLE
 & $10^{-3}$ / $500$ & $10^{-3}$ / $500$ & $10^{-3}$ / $500$  \\
\hline
\end{tabular}
\end{table}

\begin{table}[!t]
\renewcommand{\arraystretch}{1}
\centering
\footnotesize
\begin{center}
\caption{Experiment 1---Performance comparison of different methods on thermal-state tomography. }
\label{thermal_summary_full}
\begin{tabular}{|c||ccc|ccc|ccc|}
\hline
\multirow{2}{*}{Method}
 & \multicolumn{3}{c|}{$\|\wh\vrho - \vrho^\star \|_F^2/\|\vrho^\star\|_F^2$}
 & \multicolumn{3}{c|}{$\|\wh\vrho - \vrho^\star \|_1$}
 & \multicolumn{3}{c|}{$\calF(\wh\vrho, \vrho^\star)$} \\ \cline{2-10}
 & $M=5$ & $M=100$ & $M=1000$
 & $M=5$ & $M=100$ & $M=1000$
 & $M=5$ & $M=100$ & $M=1000$ \\ \hline
LR-PGD-LSE   & 0.5429 & 0.0425 & 0.0046 & 1.1997 & 0.3779 & 0.1247 & 0.5203 & 0.8308 & 0.9880 \\
LR-PM-MLE   & \textbf{0.4463} & \textbf{0.0303} & \textbf{0.0024} & \textbf{1.0924} & \textbf{0.3151} & \textbf{0.0884} & \textbf{0.5813} & \textbf{0.9227} & \textbf{0.9945} \\ \hline
LR-MPO-PGD-LSE   & 0.4239 & 0.0384 & 0.0042 & 1.0552 & 0.3526 & 0.1171 & 0.6287 & 0.9118 & 0.9897 \\
LR-MPO-PM-MLE   & \textbf{0.3707} & \textbf{0.0259} & \textbf{0.0020} & \textbf{1.0029} & \textbf{0.2911} & \textbf{0.0803} & \textbf{0.6075} & \textbf{0.9235} & \textbf{0.9954} \\ \hline
MLP-PGD-LSE          & 0.5261&	0.0426 &	0.0181 & 1.1789 &	0.3779 &	0.2456  & 0.5263 &	0.8377 &	0.9506 \\
Transformer-PGD-LSE  & \textbf{0.4113}
 &	0.0384 &	0.0062  & \textbf{1.0485}
 &	0.3495
	& 0.1414  & \textbf{0.6018
} &	0.9295
 &	0.9762 \\
MLP-PGD-MLE          & 0.4583 & 0.0225 & 0.0076 & 1.1067 & 0.2688 & 0.1551 & 0.5347 & 0.9566 & 0.9856 \\
Transformer-PGD-MLE  & 0.4267
 & \textbf{0.0129} & \textbf{0.0021} & 1.0733
 & \textbf{0.1996} & \textbf{0.0849} & 0.5778 & \textbf{0.9801} & \textbf{0.9933} \\
\hline
\end{tabular}
\end{center}
\end{table}

\begin{table}[!t]
\renewcommand{\arraystretch}{1}
\centering
\footnotesize
\begin{center}
\caption{Experiment 1---Performance comparison of different methods on GHZ-state tomography.}
\label{ghz_summary_corrected}
\begin{tabular}{|c||ccc|ccc|ccc|}
\hline
\multirow{2}{*}{Method}
 & \multicolumn{3}{c|}{$\|\wh\vrho - \vrho^\star \|_F^2/\|\vrho^\star\|_F^2$}
 & \multicolumn{3}{c|}{$\|\wh\vrho - \vrho^\star \|_1$}
 & \multicolumn{3}{c|}{$\calF(\wh\vrho, \vrho^\star)$} \\ \cline{2-10}
 & $M=5$ & $M=100$ & $M=1000$
 & $M=5$ & $M=100$ & $M=1000$
 & $M=5$ & $M=100$ & $M=1000$ \\ \hline
LR-PGD-LSE   & 0.4286 & 0.0201 & 0.0022 & 0.9259 & 0.2004 & 0.0658 & 0.7857 & 0.9900 & 0.9989 \\
LR-PM-MLE   & \textbf{0.2876} & \textbf{0.0121} & \textbf{0.0010} & \textbf{0.7584} & \textbf{0.1558} & \textbf{0.0466} & \textbf{0.8562} & \textbf{0.9939} & \textbf{0.9994} \\ \hline
LR-MPO-PGD-LSE   & 0.4074 & 0.0189 & 0.0018 & 0.9026 & 0.1945 & 0.0603 & 0.7963 & 0.9905 & 0.9991 \\
LR-MPO-PM-MLE   & \textbf{0.2659} & \textbf{0.0111} & \textbf{0.0009} & \textbf{0.7293} & \textbf{0.1491} & \textbf{0.0426} & \textbf{0.8670} & \textbf{0.9944} & \textbf{0.9995} \\ \hline
MLP-PGD-LSE          & 0.4749 &	0.0271 & 0.0023 & 0.9746&	0.2327 &	0.0684 & 0.7625 &	0.9865 &	0.9988 \\
Transformer-PGD-LSE  & 0.1191
 &	0.0084 &	0.0007  & 0.4880 &	0.1301 &	0.0366  & 0.9404 &	0.9958 &	0.9997 \\
MLP-PGD-MLE          & 0.1221 & \textbf{0.0007} & \textbf{0.0001} & 0.4943 & \textbf{0.0368} & \textbf{0.0168} & 0.9389 & \textbf{0.9997} & \textbf{0.9999} \\
Transformer-PGD-MLE  & \textbf{0.0995} & 0.0035 & 0.0003 & \textbf{0.4462} & 0.0838 & 0.0268 & \textbf{0.9502} & 0.9982 & 0.9998 \\
\hline
\end{tabular}
\end{center}
\end{table}

\begin{table}[!t]
\renewcommand{\arraystretch}{1}
\centering
\footnotesize
\begin{center}
\caption{Experiment 1---Performance comparison of different methods on zero-texture-state  tomography.}
\label{w_summary_corrected}
\begin{tabular}{|c||ccc|ccc|ccc|}
\hline
\multirow{2}{*}{Method}
 & \multicolumn{3}{c|}{$\|\wh\vrho - \vrho^\star \|_F^2/\|\vrho^\star\|_F^2$}
 & \multicolumn{3}{c|}{$\|\wh\vrho - \vrho^\star \|_1$}
 & \multicolumn{3}{c|}{$\calF(\wh\vrho, \vrho^\star)$} \\ \cline{2-10}
 & $M=5$ & $M=100$ & $M=1000$
 & $M=5$ & $M=100$ & $M=1000$
 & $M=5$ & $M=100$ & $M=1000$ \\ \hline
LR-PGD-LSE   & 0.4460 & 0.0218 & 0.0021 & 0.9445 & 0.2086 & 0.0649 & 0.7770 & 0.9891 & 0.9989 \\
LR-PM-MLE   & \textbf{0.3330} & \textbf{0.0112} & \textbf{0.0011} & \textbf{0.8161} & \textbf{0.1495} & \textbf{0.0485} & \textbf{0.8335} & \textbf{0.9944} & \textbf{0.9994} \\ \hline
LR-MPO-PGD-LSE   & 0.4110 & 0.0222 & 0.0023 & 0.9067 & 0.2141 & 0.0680 & 0.7945 & 0.9885 & 0.9988 \\
LR-MPO-PM-MLE   & \textbf{0.3009} & \textbf{0.0102} & \textbf{0.0010} & \textbf{0.7758} & \textbf{0.1425} & \textbf{0.0471} & \textbf{0.8495} & \textbf{0.9949} & \textbf{0.9994} \\ \hline
MLP-PGD-LSE          & 0.4299 &	0.0212 & 0.0018 & 0.9273 &	0.2059 &	0.0594 & 0.7850 &	0.9894 &	0.9991 \\
Transformer-PGD-LSE  & 0.1517 &	0.0057 &	0.0013  & 0.5509 &	0.1066 &	0.0517  & 0.9241 &	0.9972 &	0.9993 \\
MLP-PGD-MLE          & 0.2810 & 0.0093 & 0.0006 & 0.7497 & 0.1361 & 0.0357 & 0.8595 & 0.9954 & 0.9997 \\
Transformer-PGD-MLE  & \textbf{0.1187} & \textbf{0.0012} & \textbf{0.0002} & \textbf{0.4873} & \textbf{0.0487} & \textbf{0.0176} & \textbf{0.9406} & \textbf{0.9994} & \textbf{0.9999} \\
\hline
\end{tabular}
\end{center}
\end{table}

In the second experiment, we investigate the convergence behavior of PGD-MLE-based and PM-MLE-based methods. Specifically, we consider a thermal state at temperature $T = 0.2$ with $Q = 100$ and $M = 100$, under feasible sets $\mF\in\setF_\textup{LR},\setF_\textup{LR-MPO}$. As shown in \cref{Convergence analysis of LR MLE and PM,Convergence analysis of LR_MPO MLE and PM}, PM-MLE-based algorithms consistently exhibit faster convergence and achieve lower normalized mean-squared error and trace norm, as well as higher fidelity, compared to their PGD-MLE-based counterparts.
While the convergence rate and reconstruction accuracy of PGD-MLE-based algorithms can be improved by adopting larger step sizes, their behavior at best merely approaches that of PM and requires identifying appropriate step sizes, which often entails computationally expensive hyperparameter tuning. By contrast, PM-MLE-based algorithms naturally operate in a parameter-free manner, avoiding this additional tuning burden while achieving competitive reconstruction accuracy.

\begin{figure*}[!ht]
\centering
\subfigure{
\begin{minipage}[t]{0.31\textwidth}
\centering
\includegraphics[width=5.5cm]{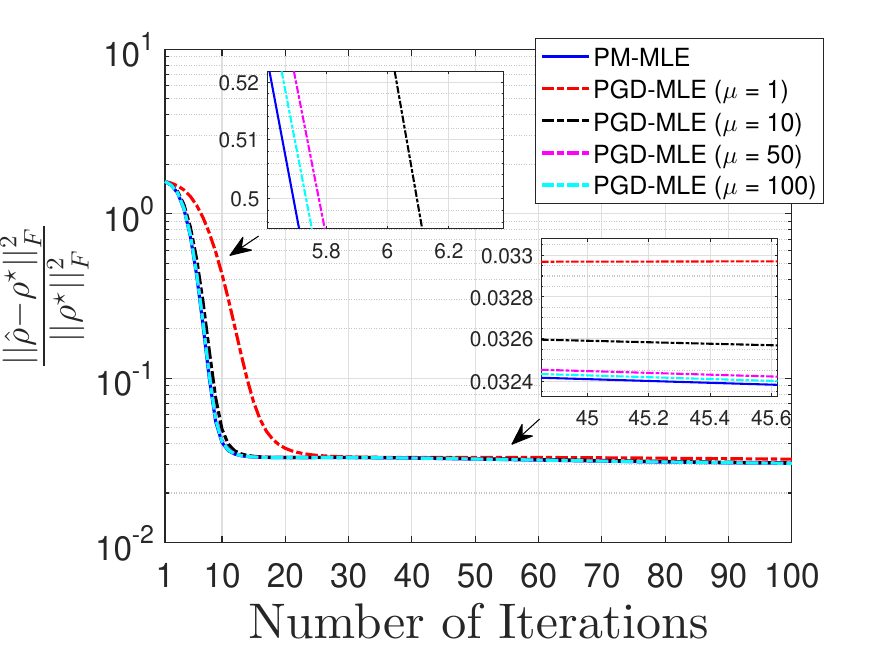}
\end{minipage}
\label{MLE_PM_LSE}
}
\subfigure{
\begin{minipage}[t]{0.31\textwidth}
\centering
\includegraphics[width=5.5cm]{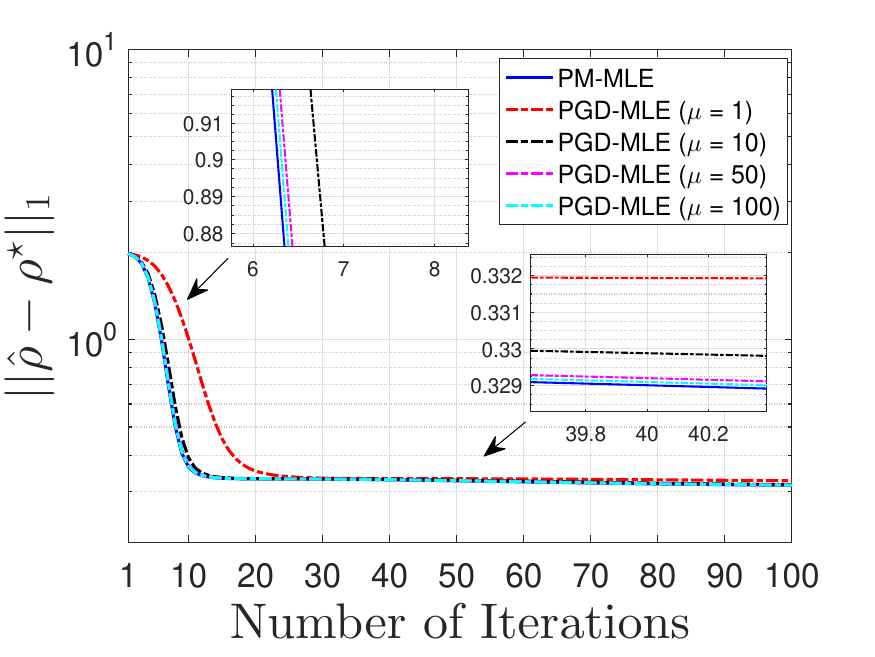}
\end{minipage}
\label{MLE_PM_trace}
}
\subfigure{
\begin{minipage}[t]{0.31\textwidth}
\centering
\includegraphics[width=5.5cm]{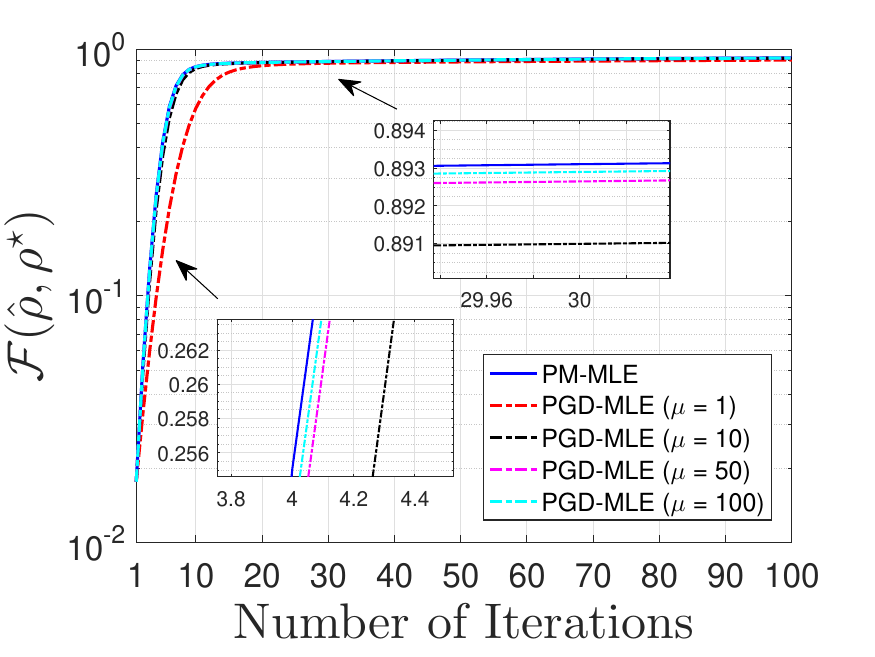}
\end{minipage}
\label{MLE_PM_fidelity}
}
\caption{Experiment 2---Convergence performance comparison between LR-PGD-MLE and LR-PM-MLE.}
\label{Convergence analysis of LR MLE and PM}
\end{figure*}

\begin{figure*}[!ht]
\centering
\subfigure{
\begin{minipage}[t]{0.31\textwidth}
\centering
\includegraphics[width=5.5cm]{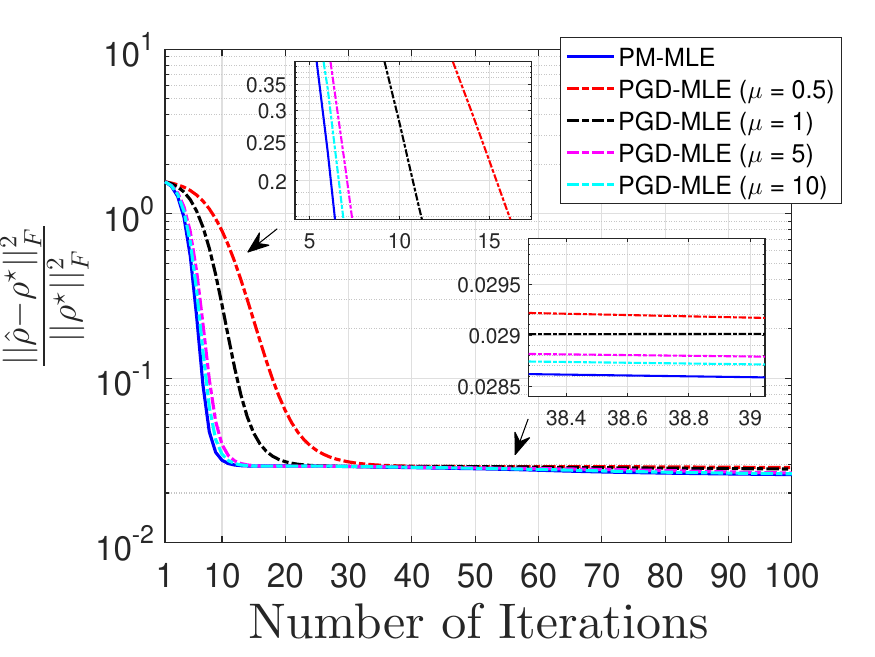}
\end{minipage}
\label{LR_MPO MLE_PM_LSE}
}
\subfigure{
\begin{minipage}[t]{0.31\textwidth}
\centering
\includegraphics[width=5.5cm]{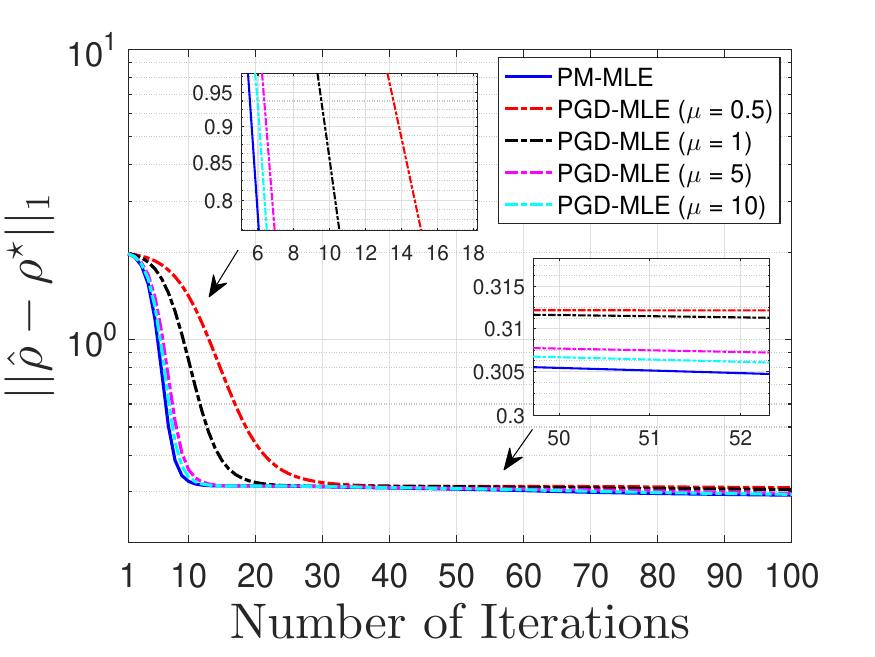}
\end{minipage}
\label{LR_MPO MLE_PM_trace}
}
\subfigure{
\begin{minipage}[t]{0.31\textwidth}
\centering
\includegraphics[width=5.5cm]{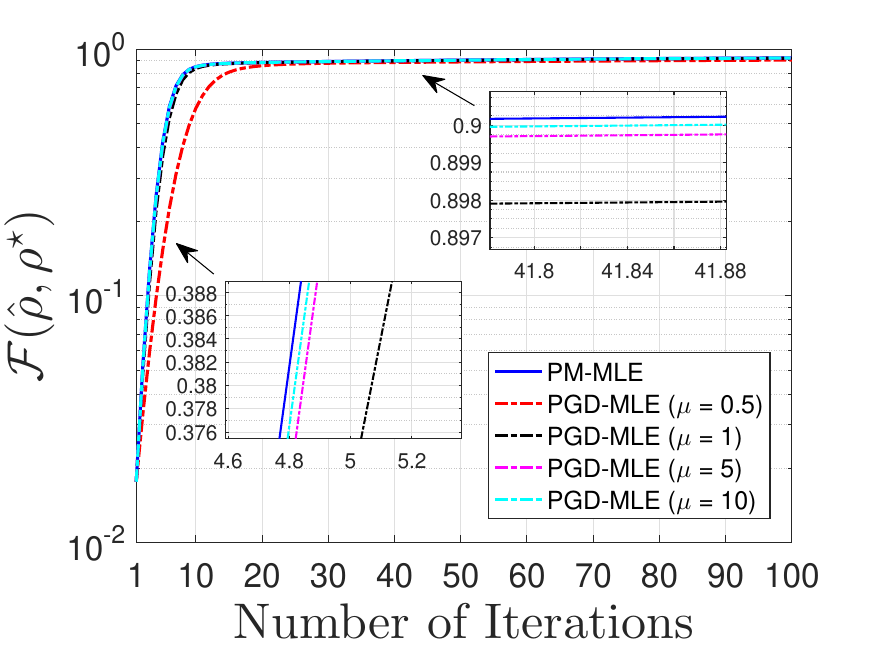}
\end{minipage}
\label{LR_MPO MLE_PM_fidelity}
}
\caption{Experiment 2---Convergence performance comparison between LR-MPO-PGD-MLE and LR-MPO-PM-MLE.}
\label{Convergence analysis of LR_MPO  MLE and PM}
\end{figure*}

Building on the preceding analysis, which establishes the strong empirical performance of PMs, we next examine their scalability with respect to the system size $n$ and the assumed low-rank parameter $r^{\textup{LR}}$. In this third experiment, the numbers of iterations are fixed to $50$ and $150$ for LR-PM-MLE and LR-MPO-PM-MLE, respectively. The larger iteration budget for LR-MPO-PM-MLE is due to the additional TT-SVD projection step required at each iteration, which introduces extra approximation error and typically leads to slower convergence compared to LR-PM-MLE.
As shown in \cref{LR_and_MPO_performance_comparison_PM_with_diff_r}, we observe a clear performance separation between the two approaches as the system size increases: LR-MPO-PM-MLE consistently outperforms LR-PM-MLE for larger $n$. This advantage can be attributed to the more favorable scaling behavior of the MPO-based formulation, whose reconstruction error grows only linearly with $n$, in contrast to the exponential scaling exhibited by the standard low-rank approach.
Moreover, both methods exhibit a pronounced dependence on the choice of $r^{\textup{LR}}$, with an optimal rank emerging that balances model expressiveness and estimation error. In particular, for low-rank quantum states, $r^{\textup{LR}} = 2$ achieves the best performance in terms of squared error, while $r^{\textup{LR}} = 1$ performs best in trace-norm and fidelity. For LR-MPO states, $r^{\textup{LR}} = 2$ yields the overall best performance in this setting.  This behavior can be explained by examining the cumulative energy captured by the leading eigenvalues of the ground-truth density matrix: the top-$1$ and top-$2$ eigenvalues account for approximately $[99.79, 98.85, 96.99, 94.49, 91.69, 88.81, 85.99, 83.27]\%$ and $[100, 100, 100, 99.98, 99.92, 99.79, 99.57, 99.24]\%$ of the total spectral mass for $n = 2, \dots, 9$, respectively. This indicates that $r^{\textup{LR}} = 2$ already captures the dominant spectral structure of the state.  As the system size increases, the many-body energy spectrum becomes increasingly dense, and the Gibbs distribution spreads over a larger number of eigenstates. As a result, a higher effective rank is required to capture the same fraction of the spectral mass, leading to the observed growth of the optimal $r^{\textup{LR}}$ with $n$.

\begin{figure*}[!tb]
\centering
\subfigure{
\begin{minipage}[t]{0.31\textwidth}
\centering
\includegraphics[width=5.5cm]{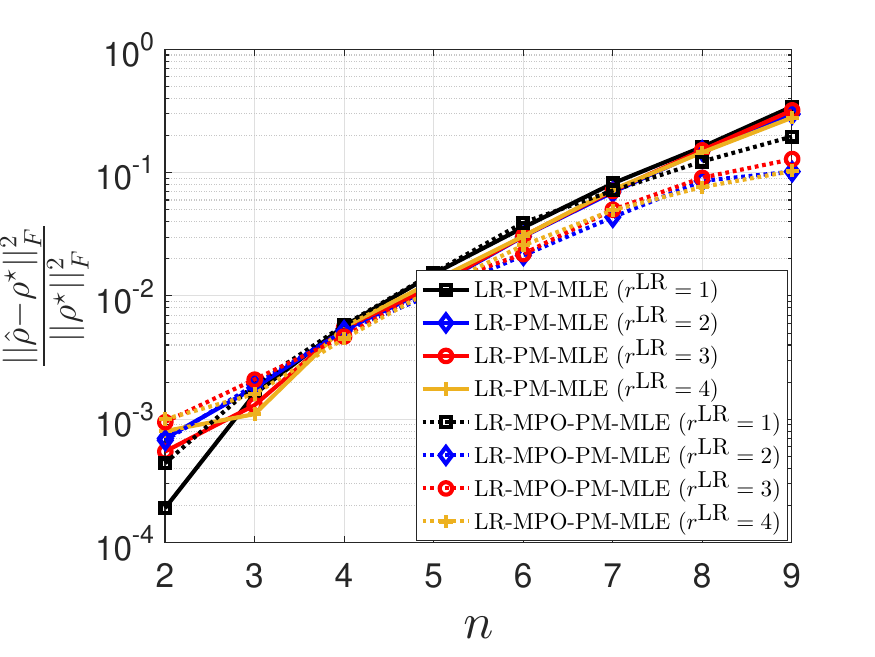}
\end{minipage}
\label{PM_LSE}
}
\subfigure{
\begin{minipage}[t]{0.31\textwidth}
\centering
\includegraphics[width=5.5cm]{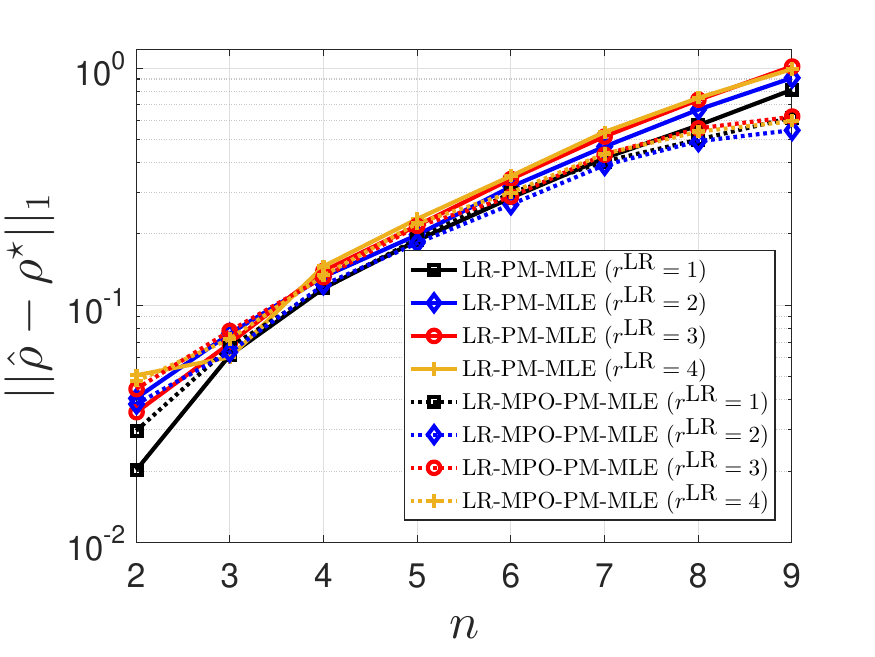}
\end{minipage}
\label{PM_trace}
}
\subfigure{
\begin{minipage}[t]{0.31\textwidth}
\centering
\includegraphics[width=5.5cm]{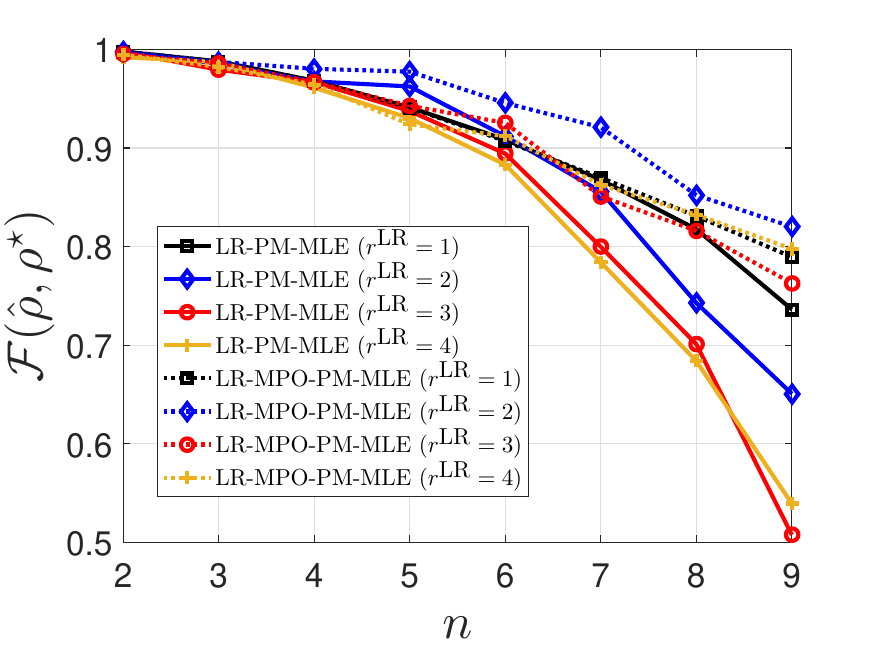}
\end{minipage}
\label{PM_fidelity}
}
\caption{Experiment 3---Performance comparison between LR-PM-MLE and MPO-PM-MLE for the thermal state at $T=0.2$, under varying system size $n$ and rank $r^{\textup{LR}}$. }
\label{LR_and_MPO_performance_comparison_PM_with_diff_r}
\end{figure*}

In the fourth experiment, we study the performance of $\mF\in\setF_\textup{NDO}$ as a function of architectural complexity, including the hidden dimension $d_{\textup{MLP}}$ for MLPs, the model dimension $d_{\textup{Trans}}$ for transformer architectures, and the network depth $L$.  We adopt the same experimental setting as in the second experiment and fix $r^{\textup{LR}} = 2$. For MLP-based models with varying $d_{\textup{MLP}}$, the step size is selected from $[0.01,\,0.03]$, and the number of iterations ranges from $300$ to $1000$, increasing with $d_{\textup{MLP}}$. When varying the depth $L$, the step size lies in $[0.01,\,0.04]$, and the number of iterations ranges from $600$ to $1200$, again increasing with $L$. For transformer-based models with different model dimensions $d_{\textup{Trans}}$, step sizes in the range $[10^{-3},\,4\times10^{-3}]$ are used, with the number of iterations varying between $400$ and $1000$ and increasing with $d_{\textup{Trans}}$. Similarly, when varying the depth $L$, the step size is chosen from $[10^{-3},\,5\times10^{-3}]$, and the number of iterations spans $600$ to $1200$, increasing monotonically with $L$. Since larger networks generally exhibit slower optimization dynamics, additional iterations are required to reach convergence. The reported iteration numbers were chosen such that the reconstruction error measures had essentially plateaued, with further training resulting in negligible changes. From \cref{MLP_Transformer_diff_d_summary,MLP_Transformer_diff_L_summary}, we observe that insufficient model capacity---manifested as overly small values of $d_{\textup{MLP}}$, $d_{\textup{Trans}}$, or $L$---leads to noticeably degraded reconstruction performance. Increasing these architectural parameters substantially improves stability as the system size $n$ grows, an effect that is most clearly reflected in the fidelity metric. Compared with the performance variation with increasing width in \cref{MLP_Transformer_diff_d_summary}, the effect of increasing depth in \cref{MLP_Transformer_diff_L_summary} is less pronounced, as $d_{\textup{MLP}} = 16$ and $d_{\textup{Trans}} = 16$ are already sufficient to achieve good reconstruction performance with shallow architectures. However, this performance gain comes at a nontrivial computational cost: larger model dimensions and deeper architectures require a significantly increased number of iterations and, consequently, longer runtime to achieve convergence.

\begin{figure*}[!tb]
\centering
\subfigure{
\begin{minipage}[t]{0.31\textwidth}
\centering
\includegraphics[width=5.5cm]{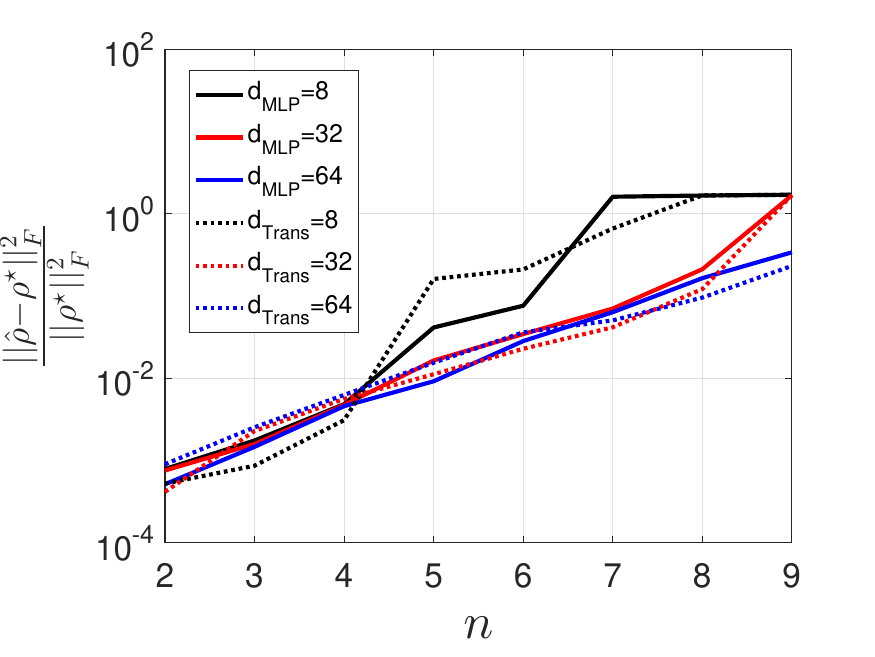}
\end{minipage}
\label{MLP Transformer_MSE_diff_d}
}
\subfigure{
\begin{minipage}[t]{0.31\textwidth}
\centering
\includegraphics[width=5.5cm]{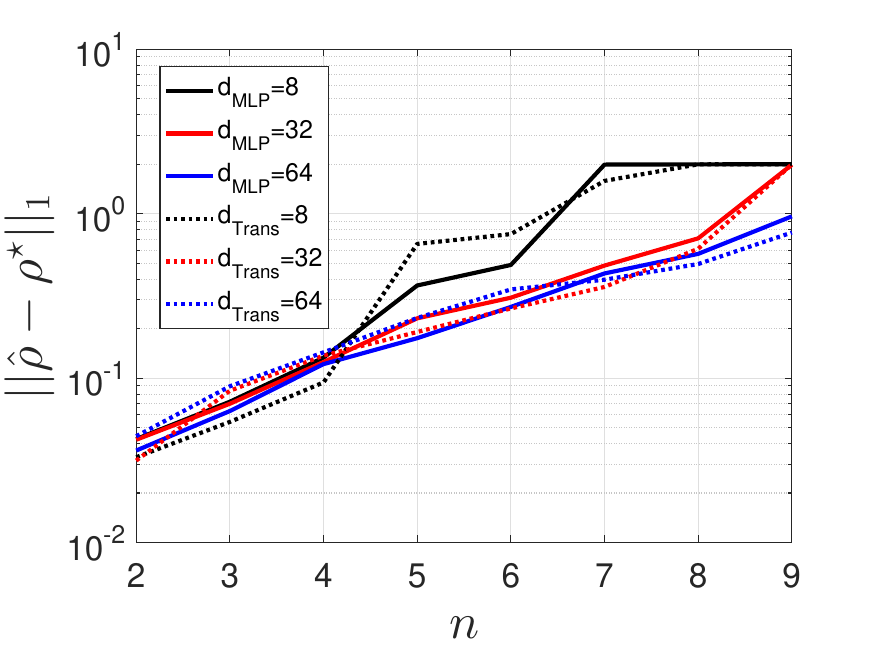}
\end{minipage}
\label{MLP Transformer_trace_diff_d}
}
\subfigure{
\begin{minipage}[t]{0.31\textwidth}
\centering
\includegraphics[width=5.5cm]{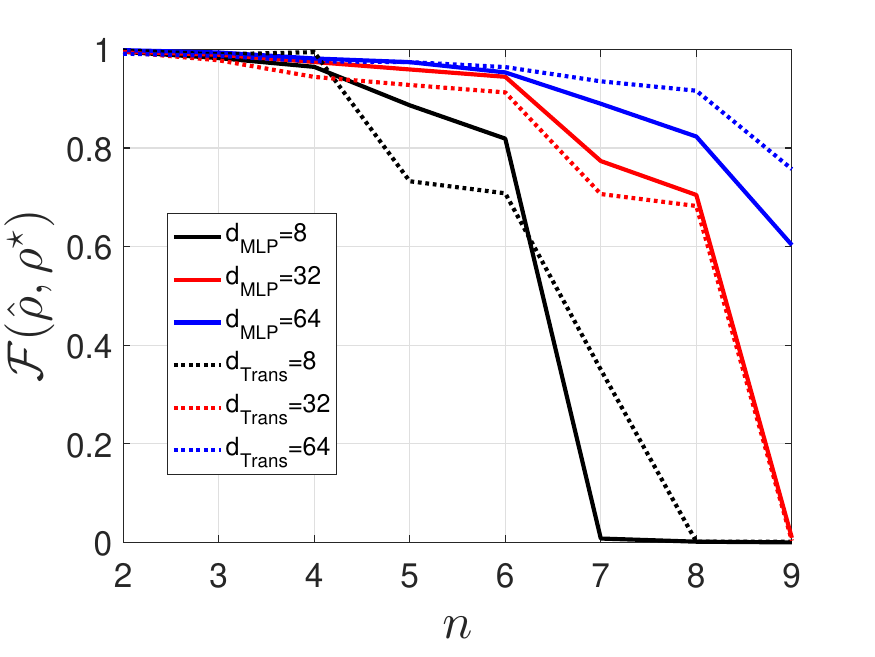}
\end{minipage}
\label{MLP Transformer_fidelity_diff_d}
}
\caption{Experiment 4---Performance comparison of MLP- and Transformer-based QST for the thermal state at $T=0.2$, under varying system size $n$, hidden dimension $d_{\textup{MLP}}$ and model dimension $d_{\textup{Trans}}$, with $r^{\textup{LR}} = 2$, $L = 2$, and $d_{\textup{MLP}}, d_{\textup{Trans}} \in \{8,32,64\}$.}
\label{MLP_Transformer_diff_d_summary}
\end{figure*}

\begin{figure*}[!t]
\centering
\subfigure{
\begin{minipage}[t]{0.31\textwidth}
\centering
\includegraphics[width=5.5cm]{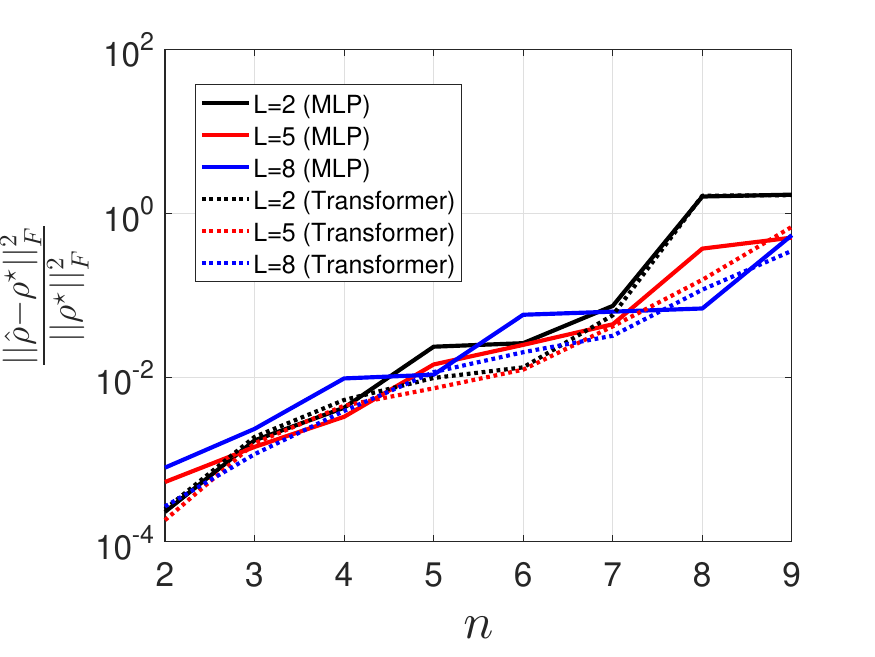}
\end{minipage}
\label{MLP Transformer_MSE_diff_L}
}
\subfigure{
\begin{minipage}[t]{0.31\textwidth}
\centering
\includegraphics[width=5.5cm]{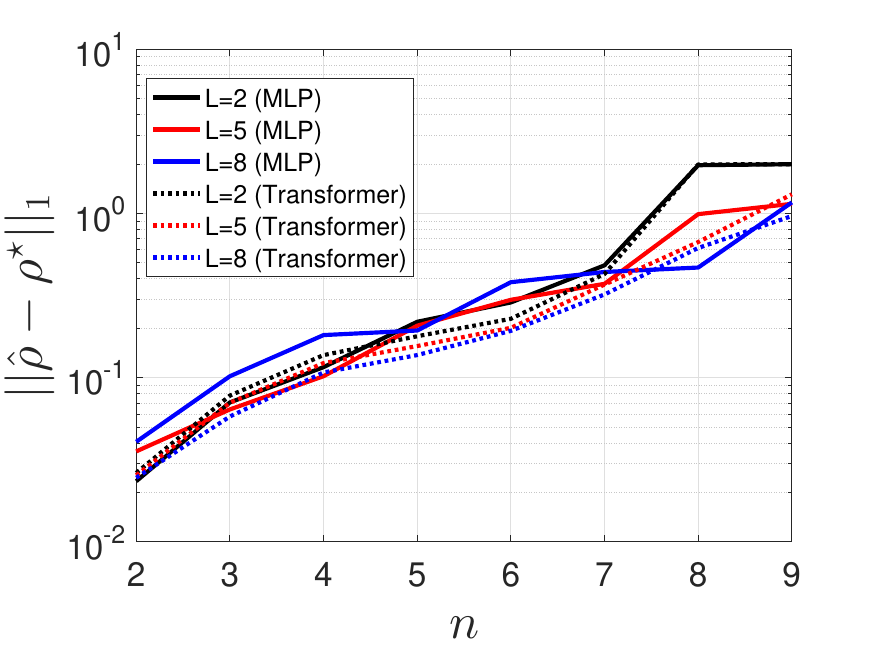}
\end{minipage}
\label{MLP Transformer_trace_diff_L}
}
\subfigure{
\begin{minipage}[t]{0.31\textwidth}
\centering
\includegraphics[width=5.5cm]{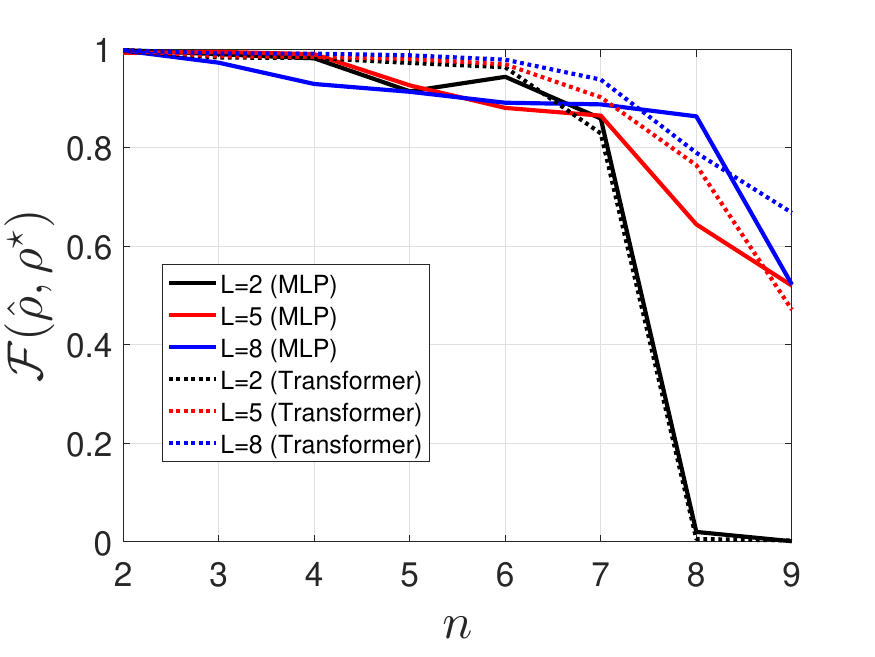}
\end{minipage}
\label{MLP_Transformer_fidelity_diff_L}
}
\caption{Experiment 4---Performance comparison of MLP- and Transformer-based QST for the thermal state at $T=0.2$, under varying system size $n$ and number of layers $L$, with $r^{\textup{LR}} = 2$, hidden dimension $d_{\textup{MLP}} = 16$ and model dimension $d_{\textup{Trans}} = 16$, for $L \in \{2,5,8\}$.}
\label{MLP_Transformer_diff_L_summary}
\end{figure*}

In the fifth experiment, we adopt the same settings as in the second experiment and fix the $r^{\textup{LR}} = 2$. For MLP-based NDOs, we set the network depth to $L = 2$ and the hidden dimension to $d_{\textup{MLP}} = 24$. As summarized in \Cref{tab:activation_comparison}, we evaluate the reconstruction performance of different activation functions, including ReLU, LeakyReLU, Tanh, Sigmoid, GELU, and SiLU. The corresponding step sizes are chosen as $0.01$, $0.01$, $0.1$, $0.1$, $0.1$, and $1$, respectively, and the number of iterations is fixed to $400$. The results indicate that ReLU and LeakyReLU achieve comparable performance and consistently outperform the other activation functions considered in this study.

\begin{table}[!ht]
\renewcommand{\arraystretch}{1}
\centering
\footnotesize
\caption{Experiment 5---Performance comparison of MLP-based QST for the thermal state at $T=0.2$, under varying activation functions, with $r^{\textup{LR}} = 2$, $L = 2$, and hidden dimension $d_{\textup{MLP}} = 24$.}
\label{tab:activation_comparison}
\begin{tabular}{|c|ccc|}
\hline
Activation & $\|\wh\vrho - \vrho^\star \|_F^2/\|\vrho^\star\|_F^2$ &  $\|\wh\vrho - \vrho^\star \|_1$ & $\calF(\wh\vrho, \vrho^\star)$ \\
\hline
ReLU       & 0.0225 & 0.2688 & 0.9566 \\
LeakyReLU & \textbf{0.0218} & \textbf{0.2649} & \textbf{0.9585} \\
Tanh      & 0.0295 & 0.3130 & 0.8880 \\
Sigmoid   & 0.0511 & 0.3919 & 0.8866 \\
GELU      & 0.0242 & 0.2778 & 0.9549 \\
SiLU      & 0.0354 & 0.3440 & 0.9000 \\
\hline
\end{tabular}
\end{table}

In the sixth and seventh experiments, we adopt the same experimental setting as in the second experiment and again fix $r^{\textup{LR}} = 2$. For NDOs based on transformer architectures, we consider multi-head self-attention with varying numbers of attention heads $M_{\textup{attn}}$ and numbers of sequence length $D_{\textup{token}}$. The network depth is set to $L = 2$, the hidden dimension is fixed to $d_{\textup{trans}} = 24$, the step size is $0.001$, and the number of iterations is $600$. Under this setting, the optimization converges for all considered configurations, and further increasing the number of iterations does not lead to any noticeable improvement in the reconstruction performance. As shown in \Cref{tab:transformer_M_comparison}, the reconstruction performance exhibits a moderate dependence on the number of attention heads. In particular, $M_{\textup{attn}} = 2$ achieves the best overall performance across the considered metrics, indicating that a moderate number of attention heads provides a favorable balance for the reconstruction task. We further investigate the effect of the attention window size $D_{\textup{token}}$ in \Cref{tab:transformer_Token_comparison}. We observe that the best performance is consistently achieved at $D_{\textup{token}} = 4$. Both smaller and larger attention window sizes lead to slightly degraded performance in terms of the considered error measures, while the variations remain relatively moderate. This suggests that the performance is not strongly sensitive to the attention window size once a minimal receptive field is included. Increasing $D_{\textup{token}}$ beyond this regime allows highly correlated tokens to interact within the same block, providing limited additional information but potentially complicating optimization, whereas an overly small attention window restricts cross-token interaction within each block. Overall, the reported best-performing configurations should be interpreted as optimal only within the range of parameters explored in this study.

\begin{table}[!ht]
\renewcommand{\arraystretch}{1}
\centering
\footnotesize
\caption{Experiment 6---Performance comparison of transformer-based QST for the thermal state at $T=0.2$, under varying $M_{\textup{attn}}$, with $r^{\textup{LR}} = 2$, $L = 2$, $D_{\textup{token}} = 4$, and model dimension $d_{\textup{trans}} = 24$.}
\label{tab:transformer_M_comparison}
\begin{tabular}{|c|ccc|}
\hline
$M_{\textup{attn}}$ & $\|\wh\vrho - \vrho^\star \|_F^2/\|\vrho^\star\|_F^2$ &  $\|\wh\vrho - \vrho^\star \|_1$ & $\calF(\wh\vrho, \vrho^\star)$ \\
\hline
1   & 0.0172 & 0.2366 & 0.9660 \\
2   & \textbf{0.0129} & \textbf{0.1996} & \textbf{0.9801} \\
4   & 0.0185 & 0.2425 & 0.9669 \\
8   & 0.0244 & 0.2811 & 0.9520 \\
24  & 0.0178 & 0.2402 & 0.9632 \\
\hline
\end{tabular}
\end{table}

\begin{table}[!ht]
\renewcommand{\arraystretch}{1}
\centering
\footnotesize
\caption{Experiment 7---Performance comparison of transformer-based QST for the thermal state at $T=0.2$, under varying $D_{\textup{token}}$, with $r^{\textup{LR}} = 2$, $L = 2$, $M_{\textup{attn}} = 2$, and model dimension $d_{\textup{trans}} = 24$.}
\label{tab:transformer_Token_comparison}
\begin{tabular}{|c|ccc|}
\hline
$D_{\textup{token}}$ & $\|\wh\vrho - \vrho^\star \|_F^2/\|\vrho^\star\|_F^2$ &  $\|\wh\vrho - \vrho^\star \|_1$ & $\calF(\wh\vrho, \vrho^\star)$ \\
\hline
1   & 0.0217 & 0.2618 & 0.9629 \\
2   & 0.0238 & 0.2780 & 0.9508 \\
4   & \textbf{0.0129} & \textbf{0.1996} & \textbf{0.9801} \\
8   & 0.0157 & 0.2254 & 0.9708 \\
16  & 0.0151 & 0.2230 & 0.9688 \\
32  & 0.0161 & 0.2270 & 0.9703 \\
64  & 0.0189 & 0.2484 & 0.9620 \\
\hline
\end{tabular}
\end{table}

In the eighth experiment, we compare the performance of Burer--Monteiro models ($\setF_\textup{simplex}$ and $\setF_\textup{LR}$) with Cholesky factorization ($\setF_\textup{Cholesky})$.
We adopt the same settings as in the first experiment and employ $Q=100$ Haar-random projective measurements with  $M = 10$ shots per POVM.
As shown in \Cref{comparison RF and CF n6_states}, this approach enables efficient optimization within the Cholesky framework. For Burer--Monteiro models, we observe that the lowest reconstruction error is achieved at $r^{\textup{LR}} = 2$ for the thermal state, whereas $r^{\textup{LR}} = 1$ yields the best performance for both the GHZ and zero-texture states. This observation is consistent with the theoretical prediction that the reconstruction error scales with the effective number of degrees of freedom.
To ensure a fair comparison in terms of model size, we also evaluate Burer--Monteiro factorization with $r^{\textup{LR}} = 32$, which has a comparable number of parameters to the Cholesky model.
The resulting reconstruction accuracy is comparable to that of the Cholesky model, and similar behavior is observed for the simplex model ($r^{\textup{LR}} = d^n= 64$). Unlike the Cholesky parametrization, however, the Burer--Monteiro framework naturally accommodates low-rank factorizations, leading to substantially reduced parameter complexity and improved scalability in structured quantum state reconstruction.


\begin{table}[!ht]
\renewcommand{\arraystretch}{1.2}
\centering
\footnotesize
\caption{Experiment 8---Performance comparison of Burer--Monteiro factorization and Cholesky factorization. The Hilbert space has dimension $d^n=2^6$. The ground-truth ranks for the thermal, GHZ, and zero-texture states are $2$, $1$, and $1$, respectively. }
\label{comparison RF and CF n6_states}
\begin{tabular}{|c||c|c|c|}
\hline
\multirow{1}{*}{Feasible Set $\mF\in\setF$}
 & \multicolumn{1}{c|}{$\|\wh\vrho - \vrho^\star \|_F^2/\|\vrho^\star\|_F^2$}
 & \multicolumn{1}{c|}{$\|\wh\vrho - \vrho^\star \|_1$}
 & \multicolumn{1}{c|}{$\calF(\wh\vrho, \vrho^\star)$} \\  \hline

\multicolumn{4}{|c|}{Thermal state} \\ \hline
$\setF_\textup{LR}$ ($r^{\textup{LR}}=2$)
 &  \textbf{0.2373}  &  \textbf{0.8202}   &  \textbf{0.6878}  \\
$\setF_\textup{LR}$  ($r^{\textup{LR}}=32$)
 &  0.2490  & 0.8753   & 0.6294   \\
 $\setF_\textup{simplex}$
 &  0.2524   &  0.8873    & 0.6119   \\
$\setF_\textup{Cholesky}$
 &   0.2538   &  0.8903   &   0.6149   \\
\hline

\multicolumn{4}{|c|}{GHZ state} \\ \hline
$\setF_\textup{LR}$ ($r^{\textup{LR}}=1$)
 & \textbf{0.1395}   &   \textbf{0.5266}   &  \textbf{0.9302}   \\
$\setF_\textup{LR}$  ($r^{\textup{LR}}=32$)
 &   0.1877   &  0.7562   &  0.7206    \\
 $\setF_\textup{simplex}$
 & 0.2057   &  0.7967     &   0.6836    \\
$\setF_\textup{Cholesky}$
 &   0.2023  &   0.7912   &  0.6902     \\ \hline

\multicolumn{4}{|c|}{Zero-texture state} \\ \hline
$\setF_\textup{LR}$ ($r^{\textup{LR}}=1$)
 & \textbf{0.1584}  &   \textbf{0.5607}   &   \textbf{0.9208}    \\
$\setF_\textup{LR}$ ($r^{\textup{LR}}=32$)
 &   0.2027   &  0.7901    &  0.7157   \\
  $\setF_\textup{simplex}$
 &  0.2035    &  0.7928    & 0.6988   \\
$\setF_\textup{Cholesky}$
 & 0.2144
  &  0.7912
   & 0.7011   \\
 \hline
\end{tabular}
\end{table}

\section{Conclusion}
In this paper, we introduced a unified, structured factorization framework for QST that integrates physical validity and structural priors within a single mathematical formulation. By parametrizing the density matrix through a constrained factorization, the proposed framework enforces physical validity together with prescribed structural priors directly within the reconstruction process, rather than relying on heuristic regularization or post hoc projection. This formulation provides a common backbone for a broad family of structured quantum state models, including low-rank states, tensor-network representations such as MPSs and LR-MPOs, as well as NDOs.

Within this unified framework, we formulated structured QST as a constrained optimization problem and studied both LSE- and MLE-based reconstruction criteria. For LSE, we established sample complexity guarantees across a wide range of structured state classes and developed geometry-aware optimization algorithms that naturally respect the imposed constraints. For MLE, we identified a distinctive structural property of the likelihood objective that enables a PM-based optimization strategy, avoiding step-size tuning while achieving stable convergence. Together, these results demonstrate that the proposed framework not only unifies previously disparate approaches to structured QST, but also leads to concrete algorithmic and theoretical advantages.


\section{Acknowledgments}
\label{sec: ack}

We acknowledge funding support from the National Science Foundation under grants CCF-2241298, ECCS-2409701, and ECCS-2540189.  ZQ gratefully acknowledges support from the MICDE Research Scholars Program at the University of Michigan. ZZ also acknowledges support from the Center for Quantum Information Science and Engineering at the Ohio State University.

\newpage

\appendices
\label{Proof of unified sample complexity}

\section{Neural Parameterizations of the Factor}
\label{Introduction of MLP and Transformer}

In this section, we introduce neural parameterizations of the factor $\mF \in \C^{d^n \times r^{\textup{LR}}}$ used in the NDO framework.  To achieve this, we adopt a unified design principle consisting of three components: $(i)$ embedding the discrete indices into a continuous representation, $(ii)$ processing the resulting representation using either a MLP or a transformer architecture, and $(iii)$ applying a linear projection to obtain the real and imaginary parts of the target entries. This formulation allows us to treat both architectures within a common functional framework while preserving their distinct inductive biases.

We now make the above pipeline explicit in terms of the input representation. For the MLP parameterization, the index tuple $(i_1,\ldots,i_n,j)$ is mapped to a deterministic feature vector
\begin{eqnarray}
\label{definition of input from MLP}
\vx = (i_1,\ldots,i_n,j)^\top \in \R^{n+1},
\end{eqnarray}
where the encoding is fixed and no learnable embedding layer is introduced.

Before introducing the transformer parameterization, we note that the representation of the discrete index set is not unique. One may tokenize individual indices, groups of indices, or even partition the entire index set into multiple shorter sequences. In this work, we adopt a joint sequence representation over the complete index set.
Specifically, we construct a sequence representation over the full index set $\{(i_1,\ldots,i_n,j)\} \in [d]^n \times [r^{\textup{LR}}]$, ordered according to a fixed lexicographic ordering.  The resulting input is given by
\begin{eqnarray}
\label{definition of input from transformer}
\mX = [\vx_1,\ldots,\vx_{d^n r^{\textup{LR}}}] \in \R^{(n+1) \times d^n r^{\textup{LR}}},
\end{eqnarray}
where each column $\vx_k = (i_1^{(k)},\ldots,i_n^{(k)},j^{(k)})$ encodes a single index tuple. In contrast to standard transformer-based NDO architectures, which process a single index tuple per token, we organize the entire index set into a sequence of tokens, with self-attention applied within local, non-overlapping blocks of this sequence (Eq.~\eqref{eq:multihead_attn}). Increasing the block size allows the model to exploit interactions among a larger number of index tuples, which leads to improved convergence behavior. In the extreme case of a single block spanning the entire sequence, the model recovers full global interactions across the index set.

Each backbone defines a mapping from an input index tuple to a complex-valued entry of the factor matrix $\mF$. Specifically, for both the MLP and transformer parameterizations, each input $\vx_k$ corresponds to a single index tuple $(i_1^{(k)},\ldots,i_n^{(k)},j^{(k)})$, and is processed by the corresponding neural network to produce a complex-valued output.
$(i)$ In the MLP parameterization, this mapping is realized by applying a feed-forward network to $\vx_k$, followed by a linear projection that outputs the real and imaginary parts of the corresponding entry of $\mF$. $(ii)$ In the transformer parameterization, the same mapping is implemented using a sequence model with self-attention layers, and each token representation is mapped via a position-shared readout network to the real and imaginary parts of the corresponding entry indexed by that token.

Thus, both parameterizations define a mapping from selected index tuples to the real and imaginary parts of the corresponding complex-valued entries of $\mF$. The resulting matrix is finally normalized to satisfy $\|\mF\|_F = 1$, which is applied uniformly across both parameterizations.

\paragraph*{Neural network models: MLP and transformer} We next formalize the architectural components underlying the two backbone parameterizations. Specifically, we define the feed-forward layers underlying the MLP, followed by the attention and two-layer MLP layers that compose a single transformer layer. These formulations make explicit how the networks map input indices to the factor $\mF$, while preserving flexibility and expressivity within the factorization framework.

\begin{itemize}
\item{\textbf{Feed-forward MLP.}} To parametrize the mapping efficiently, we adopt a standard feed-forward MLP. MLPs offer a flexible function class capable of capturing nonlinear dependencies while remaining easy to train and integrate with our parametrization framework. Given an input vector $\vx$, an $L$-layer  MLP is defined by
\begin{align}
    \label{The definition of FF MLP}
    \vh^{(0)} &= \vx, \\
    \vh^{(\ell)}
    &= \sigma_{\ell}\!\left(
            \mW^{(\ell)} \vh^{(\ell-1)} + \vb^{(\ell)}
        \right),
        \qquad \ell \in[L], \\
    \mathrm{MLP}_{\mTheta}(\vx) &= \vh^{(L)},
\end{align}
where $\mW^{(\ell)}$ and $\vb^{(\ell)}$ denote the weight matrix and bias vector of layer $\ell$, respectively, and $\sigma_\ell$ is the elementwise activation function applied at layer $\ell$ (e.g., ReLU, tanh, or sigmoid).

\item{\textbf{Transformer layer.}} A transformer layer consists of a multi-head self-attention module followed by a subsequent two-layer feed-forward network. Both components adopt a residual formulation to stabilize optimization and preserve information flow.

\begin{itemize}
\item{\textbf{Embedding (input projection) layer}.}
Each column of the input matrix
$\mX \in \R^{(n+1)\times d^n r^{\textup{LR}}}$
represents a single index tuple and lies in $\R^{n+1}$. All subsequent Transformer layers operate in a shared hidden space of dimension $N$, where $N$ denotes the model (embedding) dimension. To map the raw inputs into this space, we first apply a learnable linear embedding followed by a nonlinearity:
\begin{equation}
\label{eq:embedding}
    \hat{\mX} = \sigma_{\textup{emb}}(\mW_{\textup{emb}} \mX)
    \in \R^{N \times d^n r^{\textup{LR}}},
    \qquad \mW_{\textup{emb}} \in \R^{N \times (n+1)},
\end{equation}
where $\sigma_{\textup{emb}}$ is applied elementwise and is chosen as the ReLU activation. A fixed sinusoidal positional encoding
$\mP \in \R^{N \times d^n r^{\textup{LR}}}$ is then added to inject positional information:
$\hat{\mX} \leftarrow \hat{\mX} + \mP$.
All subsequent layers operate on the embedded representation $\hat{\mX}$.

\item{\textbf{Attention layer}.}
To control the computational cost of self-attention, we adopt a \emph{block-wise self-attention} scheme, in which self-attention is computed independently within local, non-overlapping blocks of columns of $\hat{\mX}$, rather than over the full sequence. Specifically, $\hat{\mX}$ is partitioned along its columns into contiguous blocks of size $D_{\textup{token}} \le d^n r^{\textup{LR}}$ (the \emph{attention window}), with consecutive blocks offset by a stride $ D_{\textup{token}}$, so that each column of $\hat{\mX}$ belongs to exactly one block. Let $\hat{\mX}_b \in \R^{N\times D_{\textup{token}}}$ denote one such block. Given $\hat{\mX}_b$, we define the query, key, and value projections by $\mQ, \mK, \mV \in \R^{N \times N}$, where $N$ is the hidden dimension of the transformer.
The single-head self-attention operator is defined as
\begin{equation}
\label{eq:single_head_attn}
    \textup{attn}(\hat{\mX}_b; \mQ,\mK,\mV)
    = \mV \hat{\mX}_b
    \, \sigma_{\textup{attn}}\!\left(
    \frac{(\mQ \hat{\mX}_b)^\top (\mK \hat{\mX}_b)}{\sqrt{N}}
    \right) \in \R^{N \times D_{\textup{token}}},
\end{equation}
where $\sigma_{\textup{attn}}$ denotes the softmax function applied row-wise.
For a multi-head self-attention layer with parameters
$\mTheta_{\textup{attn}} = \{ \mQ_m, \mK_m, \mV_m \}_{m\in[M_{\textup{attn}}]}$,
the outputs of all heads are aggregated with a residual connection, applied block-wise:
\begin{equation}
\label{eq:multihead_attn}
    \textup{Attn}_{\mTheta_{\textup{attn}}}(\hat{\mX}_b)
    = \hat{\mX}_b + \sum_{m=1}^{M_{\textup{attn}}}
    \textup{attn}(\hat{\mX}_b; \mQ_m, \mK_m, \mV_m)
    \in \R^{N \times D_{\textup{token}}}.
\end{equation}
The per-block outputs $\textup{Attn}_{\mTheta_{\textup{attn}}}(\hat{\mX}_b)$ are concatenated along the column dimension, in the same order as the original blocks, to recover a representation in $\R^{N\times d^n r^{\textup{LR}}}$. Setting $D_{\textup{token}} = d^n r^{\textup{LR}}$ recovers standard full self-attention over the entire sequence. The residual connection is well-defined since both terms share the same shape.

\item{\textbf{Feed-forward layer}.}
The feed-forward sublayer is a two-layer position-wise MLP, applied independently within each attention block. Given parameters
$\mW^{(1)} \in \R^{N_{\textup{ff}} \times N}$ and
$\mW^{(2)} \in \R^{N \times N_{\textup{ff}}}$,
where $N_{\textup{ff}}$ denotes the hidden width of the feed-forward network (set to $N_{\textup{ff}} = 4N$ in our implementation),
and activation function $\sigma_{\textup{ff}}$, the feed-forward layer applied to block $b$ is given by
\begin{equation}
\label{eq:ff_layer}
    \textup{FF}_{\mTheta_{\textup{ff}}}\big(\textup{Attn}_{\mTheta_{\textup{attn}}}(\hat{\mX}_b)\big)
    = \textup{Attn}_{\mTheta_{\textup{attn}}}(\hat{\mX}_b) + \mW^{(2)} \,
    \sigma_{\textup{ff}}\big(\mW^{(1)} \textup{Attn}_{\mTheta_{\textup{attn}}}(\hat{\mX}_b)\big)
    \in \R^{N \times D_{\textup{token}}},
\end{equation}
where $\mTheta_{\textup{ff}} = \{ \mW^{(1)}, \mW^{(2)} \}$. The outputs $\textup{FF}_{\mTheta_{\textup{ff}}}(\textup{Attn}_{\mTheta_{\textup{attn}}}(\hat{\mX}_b))$ from all blocks are concatenated along the column dimension, in the original block order, to produce a representation in $\R^{N\times d^n r^{\textup{LR}}}$, which serves as the input to the next transformer layer.
\end{itemize}
\end{itemize}

\section{Proof of \Cref{tab:recovery_error_structured_POVM}}
\label{proof of recovery error bound}
\begin{proof}
Let $\setX$ denote a class of quantum states induced by a factorized parametrization, i.e.,
$\setX = \{\vrho \in \mathbb{C}^{d^n \times d^n} : \vrho = \mF \mF^\dagger,\ \mF \in \setF\}$. Since every $\vrho \in \setX$ can be decomposed as $\vrho = \mF \mF^\dagger$ for some $\mF \in \setF$, it is sufficient to analyze the problem directly in terms of $\vrho$. This naturally leads to a linear inverse problem formulation in terms of the form $\langle \mA_{q,k}, \vrho \rangle$. For notational clarity, we collect the probabilities for each POVM $\{ \<\mA_{q,k}, \vrho\>\}$, into a single linear map $\calA_q:  \C^{d^n\times  d^n} \rightarrow \R^K$ of form
\begin{eqnarray}
\label{The defi of POVM element in q-th POVM}
 \calA_q(\vrho) =  \begin{bmatrix}
          \< \mA_{q,1}, \vrho  \> \\
          \vdots \\
          \< \mA_{q,K}, \vrho  \>
        \end{bmatrix}.
\end{eqnarray}
By concatenating the operators $\{\calA_q\}_{q\in[Q]}$ into a single linear map $\calA : \C^{d^n \times d^n} \to \R^{KQ}$, we obtain $KQ$ population measurements as
\begin{eqnarray}
\label{The defi of population measurement in Q cases (K measurements)}
 \vp = \calA(\vrho)  = \begin{bmatrix}
          {\bm p}_{1} \\
          \vdots \\
          {\bm p}_{Q}
        \end{bmatrix}  = \begin{bmatrix}
          \calA_1(\vrho) \\
          \vdots \\
          \calA_Q(\vrho)
        \end{bmatrix}.
\end{eqnarray}
Analogously, repeating each POVM $M$ times yields the stacked empirical frequency vector
\begin{equation}
\wh\vp = \begin{bmatrix}
    \wh \vp_1 \\ \vdots \\ \wh \vp_Q
\end{bmatrix},
\label{eq:map-M-POVM1}\end{equation}
and hence the following optimization formulation:
\begin{eqnarray}
    \label{The loss function in QST for general measurements appendix}
    \wh{\vrho} = \argmin_{\vrho\in\setX}  \frac{1}{2Q}\|\calA(\vrho) - {\widehat\vp}\|_2^2.
\end{eqnarray}
We introduce the noise vector $\veta = \wh{\vp} - \vp$.
Since $\wh{\vrho}$ is the global minimizer, the following inequality holds:
\begin{eqnarray}
    \label{whrho and rho star relationship}
    0 & \leq & \|\calA(\vrho^\star) - \wh{\vp} \|_2^2  - \|\calA(\wh{\vrho}) - \wh{\vp} \|_2^2\nonumber\\
    & = &\|\calA(\vrho^\star)-\calA(\vrho^\star) - \veta\|_2^2 - \|\calA(\wh{\vrho})-\calA(\vrho^\star) - \veta\|_2^2\nonumber\\
    & = & 2\<\calA(\vrho^\star)+\veta, \calA(\wh{\vrho} - \vrho^\star)  \> + \|\calA(\vrho^\star)\|_2^2 - \|\calA(\wh{\vrho})\|_2^2\nonumber\\
    & = & 2\<  \veta, \calA(\wh{\vrho} - \vrho^\star) \> - \|\calA(\wh{\vrho} - \vrho^\star)\|_2^2,
\end{eqnarray}
which further implies that
\begin{eqnarray}
    \label{whrho and rho^star relationship_1}
    \|\calA(\wh{\vrho} - \vrho^\star)\|_2^2 \leq 2\<  \veta, \calA(\wh{\vrho} - \vrho^\star) \>.
\end{eqnarray}
Furthermore, invoking the condition in Eq.~\eqref{requirements of second order information}, we obtain the lower bound
\begin{eqnarray}
    \label{lower bound of the left term}
    \|\calA(\wh{\vrho} - \vrho^\star)\|_2^2 \geq    C_1(Q,K)\|\wh\vrho - \vrho^\star\|_F^2.
\end{eqnarray}

To bound the right hand side of Eq.~\eqref{whrho and rho^star relationship_1}, we consider several classes of quantum states that arise from different structured parametrizations. Each class $\setX_{\textup{type}}$ denotes a set of quantum states induced by an underlying structured parametrization, namely $\setX_\textup{type}=\{\mF\mF^\dagger : \mF\in\setF_\textup{type}\}$, where $\setF_\textup{type}$ is defined in Eqs.~\eqref{simplex definition of physical quantum states}--%
\eqref{the set of LR-MPO states}. In addition to the classes obtained by the sets $\setF_\textup{type}$ defined in the main text, we introduce $\setX_{\textup{MPO}} = \{ \vrho\in\C^{d^n\times d^n}:\ \vrho = \vrho^\dagger, \trace(\vrho) = 1, \vrho(i_1 \cdots i_\nqbit, j_1\cdots j_\nqbit) = \mX_1^{i_1,j_1}  \cdots \mX_\nqbit^{i_\nqbit,j_\nqbit}, \mX_\ell^{i_\ell,j_\ell}\in\C^{r_{\ell-1}^{\textup{MPO}}\times r_\ell^{\textup{MPO}}}, \ell\in[n], r_0^{\textup{MPO}}=r_n^{\textup{MPO}}=1  \}$, which provides a structured setting that naturally enables the analysis of LR-MPO models.

\begin{itemize}
\item First focusing on $\setX_{\textup{simplex}}$, $\setX_{\textup{LR}}$ and $\setX_{\textup{MPO}}$, the cross term $\<\veta, \calA(\wh{\vrho} - \vrho^\star)\>$ can be reformulated as
    \begin{eqnarray}
    \label{upper bound of the right term}
    \<\veta, \calA(\wh{\vrho} - \vrho^\star)\> \leq \|\wh{\vrho} - \vrho^\star\|_F \cdot \max_{\vrho\in\ol\setX}\< \veta, \calA(\vrho) \>,
    \end{eqnarray}
    where the auxiliary set $\ol\setX$ is defined according to the structure of the state under consideration:
    \begin{itemize}
    \item $\ol\setX_{\textup{simplex}} = \{ \vrho\in\C^{d^n\times d^n}: \|\vrho\|_F=1, \trace(\vrho) = 0  \}$;
    \item $\ol\setX_{\textup{LR}}=  \{ \vrho\in\C^{d^n\times d^n}: \|\vrho\|_F=1, \trace(\vrho) = 0, \rank(\vrho) = 2r^{\textup{LR}}  \}$;
    \item $\ol\setX_{\textup{MPO}} =\{\vrho\in\C^{d^n\times d^n}:  \|\vrho\|_F=1, \trace(\vrho) =0,  \vrho(i_1 \cdots i_\nqbit, j_1\cdots j_\nqbit) = \mX_1^{i_1,j_1}  \cdots \mX_\nqbit^{i_\nqbit,j_\nqbit}, \mX_\ell^{i_\ell,j_\ell}\in\C^{2r_{\ell-1}^{\textup{MPO}}\times 2r_\ell^{\textup{MPO}}}, \ell\in[n]\}$.
    \end{itemize}
    Next, we apply the covering argument to bound Eq.~\eqref{upper bound of the right term}. For each structural constraint, we construct an $\epsilon$-net and analyze the corresponding covering number:
    \begin{itemize}
    \item for any fixed value of $\vrho^{(p)}\in \wt\setX_{\textup{simplex}}\subset \ol\setX_{\textup{simplex}}$, construct a $\epsilon$-net $\{\vrho^{(1)}, \dots, \vrho^{(N_{\textup{simplex}})}  \}$  such that $$\sup_{\vrho\in\ol\setX_\textup{simplex}}\min_{p\in[N_{\textup{simplex}}]}\|\vrho  - \vrho^{(p)}\|_F\leq \epsilon$$
    with covering number $N_{\textup{smplex}}\leq (\frac{9}{\epsilon})^{d^{2n}}$~\cite{candes2011tight};
    \item for any fixed value of $\vrho^{(p)}\in \wt\setX_{\textup{LR}}\subset \ol\setX_{\textup{LR}}$, construct a $\epsilon$-net $\{\vrho^{(1)}, \dots, \vrho^{(N_{\textup{LR}})}  \}$  such that $$\sup_{\vrho\in\ol\setX_{\textup{LR}}}\min_{p\in [N_{\textup{LR}}]}\|\vrho  - \vrho^{(p)}\|_F\leq \epsilon $$ with covering number $N_{\textup{LR}}\leq (\frac{9}{\epsilon})^{(2d^{n+2}+4)r^{\textup{LR}}}$~\cite{candes2011tight};
    \item for any fixed value of $\vrho^{(p)}\in \wt\setX_{\textup{MPO}}\subset \ol\setX_{\textup{MPO}}$, construct a $\epsilon$-net $\{\vrho^{(1)}, \dots, \vrho^{(N_{\textup{MPO}})}  \}$  such that $$\sup_{\vrho\in\ol\setX_{\textup{MPO}}}\min_{p\in[ N_{\textup{MPO}}]}\|\vrho  - \vrho^{(p)}\|_F\leq \epsilon $$ with covering number $N_{\textup{MPO}}\leq (\frac{4 + \epsilon}{\epsilon})^{\sum_{\ell=1}^n 4d^2r_{\ell-1}^{\textup{MPO}} r_\ell^{\textup{MPO}}  }$~\cite{qin2024quantum}.
    \end{itemize}
    Then, we can further derive
    \begin{eqnarray}
    \label{upper bound of the right term1}
    \max_{\vrho\in\ol\setX}\< \veta, \calA(\vrho) \> &\!\!\!\! = \!\!\!\!& \max_{\vrho\in\ol\setX}\sum_{q=1}^{Q}\sum_{k=1}^{K}  \< \veta_{q,k}\mA_{q,k}, \vrho - \vrho^{(p)} \> + \sum_{q=1}^{Q}\sum_{k=1}^{K}  \< \veta_{q,k}\mA_{q,k}, \vrho^{(p)} \>\nonumber\\
    &\!\!\!\! \leq \!\!\!\!&\max_{\vrho\in\ol\setX}\sum_{q=1}^{Q}\sum_{k=1}^{K} \epsilon  \left\langle \veta_{q,k}\mA_{q,k}, \frac{\vrho - \vrho^{(p)}}{\|\vrho - \vrho^{(p)}\|_F}\right\rangle + \sum_{q=1}^{Q}\sum_{k=1}^{K}  \< \veta_{q,k}\mA_{q,k}, \vrho^{(p)} \>\nonumber\\
    &\!\!\!\! \leq \!\!\!\!& \begin{dcases}
 \max_{\vrho\in\ol\setX_{\textup{simplex}}}\sum_{q=1}^{Q}\sum_{k=1}^{K} \epsilon  \< \veta_{q,k}\mA_{q,k}, \vrho \> + \sum_{q=1}^{Q}\sum_{k=1}^{K}  \< \veta_{q,k}\mA_{q,k}, \vrho^{(p)} \>, &  \setX_{\textup{simplex}} \\
 \max_{\vrho\in\ol\setX_{\textup{LR}}}\sum_{q=1}^{Q}\sum_{k=1}^{K} 2\epsilon  \< \veta_{q,k}\mA_{q,k}, \vrho \> + \sum_{q=1}^{Q}\sum_{k=1}^{K}  \< \veta_{q,k}\mA_{q,k}, \vrho^{(p)} \>, &  \setX_{\textup{LR}}  \\
 \max_{\vrho\in\ol\setX_{\textup{MPO}}}\sum_{q=1}^{Q}\sum_{k=1}^{K} n\epsilon  \< \veta_{q,k}\mA_{q,k}, \vrho \> + \sum_{q=1}^{Q}\sum_{k=1}^{K}  \< \veta_{q,k}\mA_{q,k}, \vrho^{(p)} \>, &  \setX_{\textup{MPO}}  \\
\end{dcases}
    \end{eqnarray}
    where the second inequality follows \cite[Eq. (90)]{qin2024quantum}. Finally, by choosing
    \begin{eqnarray}
    \label{upper bound of the right term pre}
    \epsilon = \begin{cases}
 \frac{1}{2}, &  \setX_{\textup{simplex}} \\
 \frac{1}{4}, &  \setX_{\textup{LR}}  \\
 \frac{1}{2n}, &  \setX_{\textup{MPO}}  \\
\end{cases},
    \end{eqnarray}
    we obtain the desired bounds for each structured set:
    \begin{eqnarray}
    \label{upper bound of the right term2}
    \max_{\vrho\in\ol\setX}\sum_{q=1}^{Q}\sum_{k=1}^{K}  \< \veta_{q,k}\mA_{q,k}, \vrho  \>\leq 2\sum_{q=1}^{Q}\sum_{k=1}^{K}  \< \veta_{q,k}\mA_{q,k}, \vrho^{(p)} \>.
    \end{eqnarray}
    We consider any fixed value of $\vrho^{(p)}$ and apply \cite[Lemma 14]{qin2024quantum} to establish a concentration inequality for the expression $\sum_{q=1}^{Q}\sum_{k=1}^{K}  \< \veta_{q,k}\mA_{q,k}, \vrho^{(p)} \>$. Using Eq.~\eqref{requirements of third order information}, we have
    \begin{eqnarray}
    \label{concentration inequality of fixed quantum state}
    \P{\sum_{q=1}^{Q}\sum_{k=1}^{K}  \< \veta_{q,k}\mA_{q,k}, \vrho^{(p)} \>> t} \leq 2e^{-\frac{Mt^2}{C C_2(Q,K)}},
    \end{eqnarray}
    where $C$ denotes a universal constant. Combining Eq.~\eqref{upper bound of the right term2} with Eq.~\eqref{concentration inequality of fixed quantum state} gives
    \begin{eqnarray}
    \label{concentration inequality of fixed quantum state1}
    \P{\max_{\vrho\in\ol\setX}\sum_{q=1}^{Q}\sum_{k=1}^{K}  \< \veta_{q,k}\mA_{q,k}, \vrho  \> > t}
    &\!\!\!\!\leq \!\!\!\!& \P{\sum_{q=1}^{Q}\sum_{k=1}^{K}  \< \veta_{q,k}\mA_{q,k}, \vrho^{(p)} \> > \frac{t}{2}} \nonumber\\
    &\!\!\!\!\leq \!\!\!\!& \begin{cases}
    2\left(\frac{9}{\epsilon}\right)^{d^{2n}}e^{-\frac{Mt^2}{C C_2(Q,K)}}, &  \setX_{\textup{simplex}} \\
    2\left(\frac{9}{\epsilon}\right)^{(2d^{n+2}+4)r^{\textup{LR}}}e^{-\frac{Mt^2}{C C_2(Q,K)}}, &  \setX_{\textup{LR}}  \\
    2\left(\frac{4 + \epsilon}{\epsilon}\right)^{\sum_{\ell=1}^n 4d^2r_{\ell-1}^{\textup{MPO}} r_\ell^{\textup{MPO}}  }e^{-\frac{Mt^2}{C C_2(Q,K)}}, &  \setX_{\textup{MPO}}  \\
    \end{cases}\nonumber\\
    &\!\!\!\!\leq \!\!\!\!& \begin{cases}
    e^{-\frac{Mt^2}{C C_2(Q,K)}+O(d^{2n})}, &  \setX_{\textup{simplex}} \\
    e^{-\frac{Mt^2}{C C_2(Q,K)} + O(d^nr^{\textup{LR}})}, &  \setX_{\textup{LR}}  \\
    e^{-\frac{Mt^2}{C C_2(Q,K)} + O(\sum_{\ell=1}^n d^2r_{\ell-1}^{\textup{MPO}} r_\ell^{\textup{MPO}} \log n)}, &  \setX_{\textup{MPO}}  \\
    \end{cases}.
    \end{eqnarray}
    Through taking
    $$\hat t = \begin{cases}
    O(\sqrt{C_2(Q,K)d^{2n}/M}), &  \setX_{\textup{simplex}} \\
    O(\sqrt{C_2(Q,K)d^{n}r^{\textup{LR}}/M}), &  \setX_{\textup{LR}}  \\
    O(\sqrt{C_2(Q,K)\sum_{\ell=1}^n d^2r_{\ell-1}^{\textup{MPO}} r_\ell^{\textup{MPO}} \log n/M}), &  \setX_{\textup{MPO}}  \\
    \end{cases},$$
    we further obtain
    \begin{eqnarray}
    \label{concentration inequality of fixed quantum state2}
    \P{\max_{\vrho\in\ol\setX}\sum_{q=1}^{Q}\sum_{k=1}^{K}  \< \veta_{q,k}\mA_{q,k}, \vrho  \> \leq \hat t} \geq   \begin{cases}
    1 - e^{-\Omega(d^{2n})}, &  \setX_{\textup{simplex}} \\
    1 - e^{-\Omega(d^nr^{\textup{LR}})}, &  \setX_{\textup{LR}}  \\
    1 - e^{-\Omega(\sum_{\ell=1}^n d^2r_{\ell-1}^{\textup{MPO}} r_\ell^{\textup{MPO}} \log n)}, &  \setX_{\textup{MPO}}  \\
    \end{cases}
    \end{eqnarray}
    Hence, we have
    \begin{eqnarray}
    \label{upper bound of the right term final}
    \<\veta, \calA(\wh{\vrho} - \vrho^\star)\>\leq   \begin{cases}
    O\bigg(\sqrt{\frac{C_2(Q,K)d^{2n}}{M}}\|\wh{\vrho} - \vrho^\star\|_F\bigg), &  \setX_{\textup{simplex}} \\
    O\bigg(\sqrt{\frac{C_2(Q,K)d^{n}r^{\textup{LR}}}{M}}\|\wh{\vrho} - \vrho^\star\|_F\bigg), &  \setX_{\textup{LR}}  \\
    O\bigg(\sqrt{\frac{C_2(Q,K)\sum_{\ell=1}^n d^2r_{\ell-1}^{\textup{MPO}} r_\ell^{\textup{MPO}} \log n}{M}}\|\wh{\vrho} - \vrho^\star\|_F\bigg), &  \setX_{\textup{MPO}}  \\
    \end{cases}.
    \end{eqnarray}
Specifically, note that $\setX_{\textup{Cholesky}}\subset \setX_{\textup{simplex}}$. Therefore, for $\wh{\vrho}, \vrho^\star \in \setX_{\textup{Cholesky}}$, we have
\begin{eqnarray}
    \label{upper bound of the right term final Cholesky}
    \<\veta, \calA(\wh{\vrho} - \vrho^\star)\>\leq O\bigg(\sqrt{\frac{C_2(Q,K)d^{2n}}{M}}\|\wh{\vrho} - \vrho^\star\|_F\bigg), & \setX_{\textup{Cholesky}}.
\end{eqnarray}

\item Next, we turn to the set $\setX_{\textup{LR-MPO}} $.  Since any $\vrho = \mF\mF^\dagger$ can be regarded as an MPO state with MPO ranks $[(r_1^{\textup{LR-MPO}})^2, \dots, (r_{n-1}^{\textup{LR-MPO}})^2 ]$, the preceding analysis applies directly. In particular, with probability at least $1 - e^{-\Omega(\sum_{\ell=1}^{n} d^2 (r_{\ell-1}^{\textup{LR-MPO}})^2 (r_\ell^{\textup{LR-MPO}})^2 \log n)}$, we have
    \begin{eqnarray}
    \label{upper bound of the right term final LR-MPO}
    \<\veta, \calA(\wh{\vrho} - \vrho^\star)\>\leq O\left(\sqrt{\frac{C_2(Q,K)\sum_{\ell=1}^n d^2 (r_{\ell-1}^{\textup{LR-MPO}})^2 (r_\ell^{\textup{LR-MPO}})^2 \log n}{M}}\|\wh{\vrho} - \vrho^\star\|_F\right), \quad\setX_\textup{LR-MPO}.
    \end{eqnarray}

\item Finally, we turn our attention to the set $\setX_{\textup{MPS}} $. Although MPS can be viewed as a special case of MPO, we aim to derive a tighter bound rather than directly specializing the MPO result. To this end, we establish a connection between $\<\veta, \calA(\wh{\vrho} - \vrho^\star)\>$ and $\|\wh\vf - \vf^\star\|_2$. Specifically, the cross term $\<\veta, \calA(\wh{\vrho} - \vrho^\star)\>$ can be reformulated as
    \begin{eqnarray}
    \label{upper bound of the right term MPS}
    \<\veta, \calA(\wh{\vrho} - \vrho^\star)\> \leq \|\wh\vf - \vf^\star\|_2 \cdot \max_{\vf\in\ol\setX_{\textup{MPS}}} (\< \veta, \calA(\wh\vf\vf^\dagger) \> + \< \veta, \calA(\vf{\vf^\star}^\dagger) \>),
    \end{eqnarray}
    where the set $\ol\setX_{\textup{MPS}}$ is defined as $\ol \setX_{\textup{MPS}}=  \{ \vf\in\C^{d^n\times 1}:    \|\vf\|_2=1,  \vf(i_1 \cdots i_\nqbit) = \mX_1^{i_1}  \cdots \mX_\nqbit^{i_\nqbit},\mX_\ell^{i_\ell}\in\C^{2r_{\ell-1}^{\textup{MPS}}\times 2r_\ell^{\textup{MPS}}},\ell=1,\dots,n \}$. Notice that $\|\wh\vf\vf^\dagger\|_F\leq 1$ and $\|\vf{\vf^\star}^\dagger\|_F\leq 1$. Following the analysis for MPO states, we obtain
    \begin{eqnarray}
    \label{upper bound of the right term final MPS}
    \<\veta, \calA(\wh{\vrho} - \vrho^\star)\>\leq  O\left(\sqrt{\frac{C_2(Q,K)\sum_{\ell=1}^n dr_{\ell-1}^{\textup{MPS}} r_\ell^{\textup{MPS}} \log n}{M}}\|\wh\vf - \vf^\star\|_F\right),
    \end{eqnarray}
    with probability $1 - e^{-\Omega(\sum_{\ell=1}^n d r_{\ell-1}^{\textup{MPS}} r_\ell^{\textup{MPS}} \log n)}$. Moreover, applying the bound $\|\wh\vf - \vf^\star\|_F\leq \frac{1}{2(\sqrt{2}-1)}\|\wh{\vrho} - \vrho^\star\|_F$ \cite[Lemma 41]{ge2017no}, we further deduce
      \begin{eqnarray}
    \label{upper bound of the right term final MPS1}
    \<\veta, \calA(\wh{\vrho} - \vrho^\star)\>\leq  O\left(\sqrt{\frac{C_2(Q,K)\sum_{\ell=1}^n dr_{\ell-1}^{\textup{MPS}} r_\ell^{\textup{MPS}} \log n}{M}}\|\wh{\vrho} - \vrho^\star\|_F\right).
    \end{eqnarray}

\end{itemize}

By combining Eq.~\eqref{lower bound of the left term} with the preceding analysis, we arrive at the following bound
    \begin{eqnarray}
    \label{recovery error F norm final}
    \|\wh{\vrho} - \vrho^\star\|_F\leq   \begin{cases}
    O\bigg(\sqrt{\frac{C_2(Q,K)d^{2n}}{C_1^2(Q,K)M}} \bigg), &  \setX_{\textup{simplex}},\setX_{\textup{Cholesky}} \\
    O\bigg(\sqrt{\frac{C_2(Q,K)d^{n}r^{\textup{LR}}}{C_1^2(Q,K)M}} \bigg), &  \setX_{\textup{LR}} \\
    O\bigg(\sqrt{\frac{C_2(Q,K)\sum_{\ell=1}^n dr_{\ell-1}^{\textup{MPS}} r_\ell^{\textup{MPS}} \log n}{C_1^2(Q,K)M}} \bigg), &  \setX_{\textup{MPS}}  \\
    O\bigg(\sqrt{\frac{C_2(Q,K)\sum_{\ell=1}^n d^2 (r_{\ell-1}^{\textup{LR-MPO}})^2 (r_\ell^{\textup{LR-MPO}})^2 \log n}{C_1^2(Q,K)M}} \bigg), &  \setX_{\textup{LR-MPO}}  \\
    \end{cases}.
    \end{eqnarray}

\end{proof}


\end{document}